\documentclass[journal]{IEEEtran}
\usepackage{amsmath, amssymb, amsthm}
\usepackage{graphicx}
\usepackage{booktabs}
\usepackage{multirow}
\usepackage{algorithm}
\usepackage{algorithmic}
\usepackage{caption}
\usepackage{subcaption}
\usepackage{enumitem}
\usepackage{times}
\usepackage{url}
\usepackage{float}
\usepackage{setspace}
\usepackage{titlesec}
\usepackage{xcolor}
\usepackage{acronym}
\usepackage{marvosym}
\usepackage{cite} 
\usepackage[colorlinks,
            linkcolor=blue,
            anchorcolor=blue,
            citecolor=blue]{hyperref}
            
\newtheorem{theorem}{Theorem}[section]   
\usepackage{orcidlink}

\ifCLASSINFOpdf
\else
\fi

\acrodef{PDEs}{Partial Differential Equations}
\acrodef{LR-QA}{Low-Rank Query Attention}
\acrodef{FNO}{Fourier Neural Operator}
\acrodef{KAN}{Kolmogorov–Arnold Network}
\acrodef{LSFs}{Local Spatial Features}
\acrodef{WFNO}{Weighted Fourier Neural Operator}
\acrodef{SIREN}{Sinusoidal Representation Network}
\acrodef{FVGN}{finite volume informed neural network}
\acrodef{Q-former}{Querying Transformer}
\acrodef{CFD}{Computational Fluid Mechanics}
\acrodef{LRQ-Solver}{Low-Rank Query-based PDE Solver}
\acrodef{PCLM}{Parameter-Conditioned Lagrangian Modeling}
\acrodef{PCE}{Parameter-Conditioned Encoder}
\acrodef{MLP}{Multilayer Perceptron}
\acrodef{FEM}{Finite Element Method}
\acrodef{CAE}{Computer-Aided Engineering}
\acrodef{CAD}{Computer-Aided Design}
\begin{document}

\title{LRQ-Solver: A Transformer-Based Neural Operator for Fast and Accurate Solving of Large-scale 3D PDEs}


\author{
\IEEEauthorblockN{
Peijian~Zeng$^{}$\orcidlink{0009-0009-3352-3565}\textsuperscript{*}, 
Guan~Wang$^{}$\textsuperscript{*}, 
Haohao~Gu$^{}$, 
Xiaoguang~Hu$^{}$,
Tiezhu~Gao$^{}$,\\
Zhuowei~Wang$^{}$\orcidlink{0000-0001-6479-5154}\textsuperscript{\Letter}, 
Aimin~Yang$^{}$\orcidlink{0009-0002-1751-4801}\textsuperscript{\Letter}, 
Xiaoyu~Song$^{}$\orcidlink{0000-0002-6583-9400}, 
}\\

\thanks{This work was supported in part by the Key-Area Research and Development Program of Guangdong Province under Grant 2021B0101190004 and 2021B0101190003, and in part by the National Natural Science Foundation of China under Grant 62472106.}
\thanks{P. Zeng, Z. Wang, and A. Yang are with the School of Computer Science and Technology, Guangdong University of Technology, Guangzhou 510000, China (e-mail: {lil\_ken@163.com , zwwang@gdut.edu.cn , amyang18@163.com }).}

\thanks{G. Wang, H. Gu, X. Hu, and T. Gao are with Baidu Inc., Beijing 100085, China (e-mail: {wangguan12@baidu.com , guhaohao@baidu.com , huxiaoguang@baidu.com , gaotiezhu@baidu.com }).}

\thanks{X. Song is with the Department of Electrical and Computer Engineering, Portland State University, Portland, OR 97207, USA (e-mail: songx@pdx.edu ).}

\thanks{\textsuperscript{*} Peijian~Zeng and Guan~Wang are co-first authors.}
\thanks{\textsuperscript{\Letter} Zhuowei~Wang and Aimin~Yang are co-corresponding authors.}
}

\maketitle

\begin{abstract}
Solving large-scale \ac{PDEs} on complex three-dimensional geometries represents a central challenge in scientific and engineering computing, often impeded by expensive pre-processing stages and substantial computational overhead. We present \ac{LRQ-Solver}, a physics-integrated deep learning framework for fast, accurate, and scalable \ac{CAE} simulations of complex three-dimensional geometries in integrated circuit and system design. Built upon the \ac{PCLM} that embeds physical consistency into the learning process and the \ac{LR-QA} module that reduces attention complexity from $O(N^2)$ to $O(NC^2 + C^3)$ via covariance decomposition, \ac{LRQ-Solver} enables efficient multi-configuration analysis directly within \ac{CAD}-driven workflows. Evaluated on industrial benchmarks, it achieves a 38.9\% error reduction on DrivAerNet++ and 28.76\% on the 3D Beam dataset, with up to $50\times$ training speedup and support for simulations with 2 million points on a single GPU. By accelerating \ac{PDEs}-based \ac{CAE} tasks—such as thermal, mechanical, or electromagnetic analysis—\ac{LRQ-Solver} enhances the responsiveness and scalability of design automation pipelines. Code to reproduce the experiments is available at \href{https://github.com/LilaKen/LRQ-Solver}{https://github.com/LilaKen/LRQ-Solver}
\end{abstract}

\section{Introduction}

Physical phenomena across natural and industrial systems—from solar dynamo cycles to aircraft aerodynamics—are universally governed by \ac{PDEs}~\cite{vasil2024solar,santos2025self}. In engineering, critical applications such as vehicle aerodynamics and structural stress analysis rely fundamentally on PDE-based modeling~\cite{basar2001nonlinear,evans2022partial}. Accurate and efficient PDE solutions are indispensable for predicting complex nonlinear systems—from weather forecasting to nuclear simulations~\cite{spears2025predicting}—and for optimizing industrial designs. In \ac{CAD} workflows, where geometry is often parameterized and subject to frequent modifications, repeated high-fidelity \ac{PDEs} solves are required to evaluate performance across design configurations—posing a major bottleneck in design automation and rapid prototyping~\cite{Grossmann2007,solin2005partial}. However, analytical solutions remain intractable for most real-world problems, forcing reliance on numerical discretization methods that suffer from high computational cost, mesh generation overhead, and sensitivity to geometric complexity~\cite{solin2005partial,Grossmann2007}. Neural \ac{PDEs} solvers have emerged as a transformative alternative, learning operators from simulation data to deliver mesh-free, resolution-independent predictions in seconds~\cite{li2021fourier,gino,li2025geometric}, thereby enabling tight integration with \ac{CAD} environments and accelerating design-space exploration without repeated meshing or solver setup.

Despite their promise, neural solvers face two fundamental limitations in industrial deployment. \textbf{First, an accuracy bottleneck}: most architectures decouple global design parameters from local physical dynamics, failing to capture the interdependence between the system-level constraints, such as chassis length or material thickness—and field behavior. For instance, an A-pillar may exhibit benign flow separation for a compact car but trigger strong vortices for an extended wheelbase; similarly, a B-pillar fillet that evenly distributes stress at nominal thickness becomes a stress concentrator when thinned. Existing fusion strategies—e.g., GNOT’s point-wise embeddings~\cite{pmlr-v202-hao23c}, GINOT’s feature concatenation~\cite{liu2025geometry}, or Geom-DeepONet’s multiplicative fusion~\cite{geomdeeponet}—lack explicitly physical coupling, resulting in black-box parameter sensitivity and inconsistent generalization. \textbf{Second, an efficiency bottleneck}: scaling to million-point geometries is hindered by $O(N^2)$ attention complexity. Methods like Transolver++~\cite{luo2025transolverplus} rely on per-point clustering with linear overhead, rendering them infeasible for industrial-scale point clouds.

\begin{figure}[htbp]
    \centering
    \includegraphics[width=0.5\textwidth]{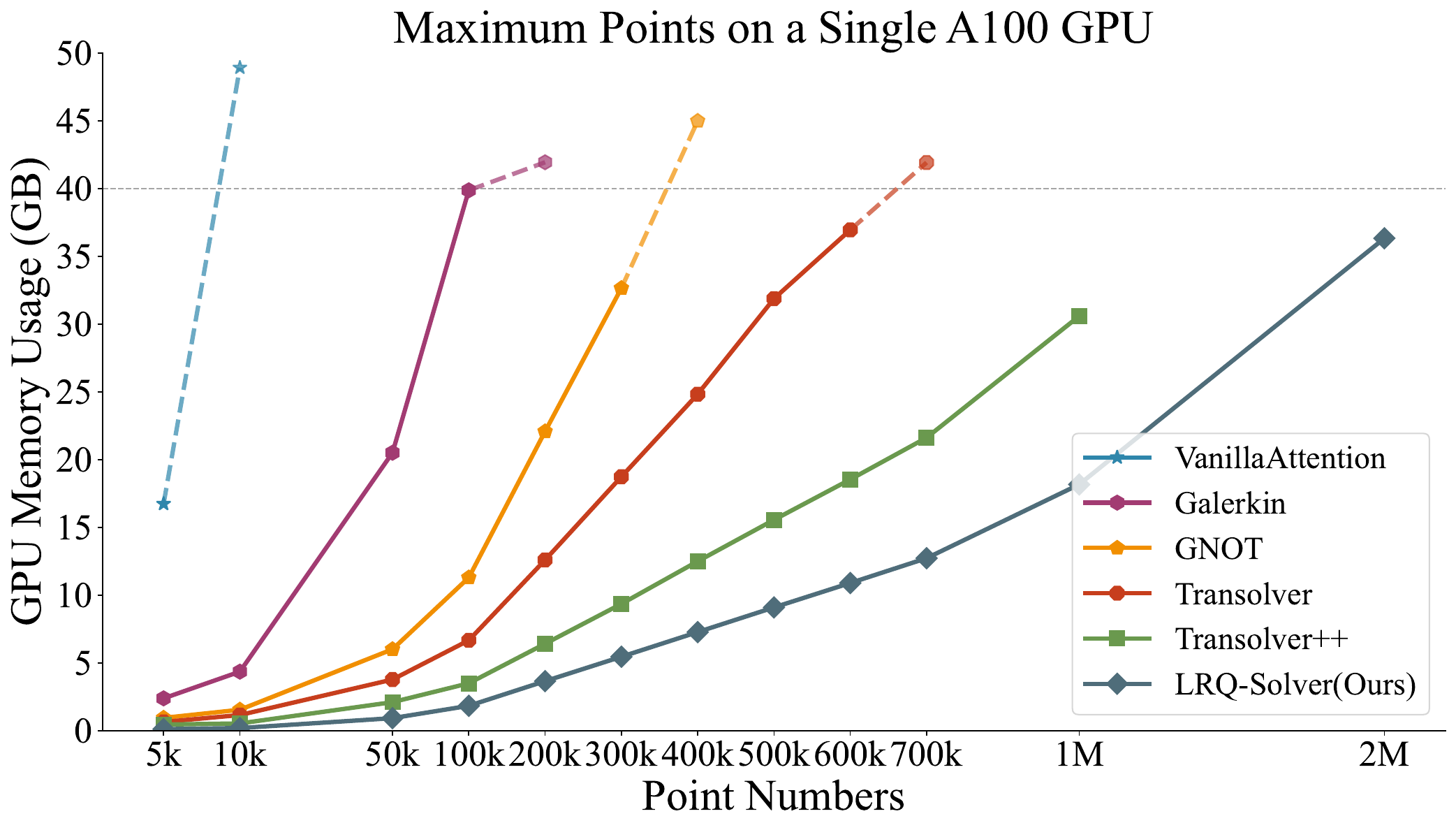}
    \caption{Comparison of model capability in handling large geometries.}
    \label{fig:model_performance_gpu_use}
\end{figure}

To overcome these dual challenges, we propose \ac{LRQ-Solver}, a unified physics-integrated framework comprising two synergistic innovations. For \textbf{accuracy}, we introduce the \ac{PCLM}, which explicitly models each material point’s state as a joint function of its spatial coordinate and a global shape descriptor. Through a \ac{PCE}, high-level design parameters are mapped into a latent control field that modulates the entire physical domain, embedding design context directly into field evolution—rather than via concatenation or attention. This establishes a structured, physics-informed mapping from geometry to response, ensuring consistent, interpretable, and generalizable predictions across configurations. For \textbf{efficiency}, we develop the \ac{LR-QA}, which exploits second-order field statistics (e.g., velocity-velocity or stress-strain correlations) to construct a global coherence kernel. A single covariance decomposition compresses $N$ points into $C \ll N$ coherent structures, reducing attention complexity from $O(N^2)$ to $O(NC^2 + C^3)$, and enabling end-to-end training on point clouds up to 2 million points using a single A100 GPU—doubling prior capacity. As shown in Fig.~\ref{fig:model_performance_gpu_use}, this breakthrough sets a new standard for scalability. Together, \ac{PCLM} and \ac{LR-QA} form a differentiable framework that integrates pseudo-physics fields with control volume integrals derived from conservation laws, replacing direct regression with physically grounded operations to ensure consistency and enable gradient-based design optimization at industrial scale.

The main contributions of this work are:

\begin{enumerate}
    \item \textbf{A physics-integrated, end-to-end differentiable framework} that unifies global design control with local field evolution through conservation-law-aware operations. By embedding pseudo-physics fields with differentiable control volume integrals, our framework ensures physically consistent system-level responses across variable configurations—enabling robust gradient-based optimization for industrial-scale design tasks under strong, nonlinear physical fields.

    \item \textbf{High-accuracy design-aware modeling via \ac{PCLM}}, which explicitly couples local material states with global shape parameters through a \ac{PCE}-driven latent control field. Unlike black-box parameter injection methods, PCLM establishes a structured, interpretable mapping from design space to physical response, achieving superior generalization in multi-configuration scenarios—e.g., 38.8\% MSE reduction on DrivAerNet++~\cite{drivaernetplus} (MSE=5.56) and 28.8\% on 3D Beam (MSE=1.66)—demonstrating unprecedented fidelity in capturing geometry-modulated physical behavior.

    \item \textbf{High-efficiency large-scale simulation via \ac{LR-QA}}, which exploits the low-rank, long-range-correlated structure of physical fields to replace point-wise attention with global coherence derived from second-order statistics. This reduces complexity from $O(N^2)$ to $O(NC^2 + C^3)$, enabling real-time inference at \textbf{0.005 seconds} on point clouds up to \textbf{2 million points} using a single A100 GPU—setting a new standard for scalability without sacrificing resolution or geometric fidelity.
\end{enumerate}

Experimental results show that \ac{LRQ-Solver} not only surpasses existing neural PDE solvers in accuracy and efficiency but also successfully bridges the gap between data-driven modeling and industrial-scale, multi-configuration engineering simulation—paving the way for fast, accurate, and physically consistent AI systems in real-world design workflows.

\section{Related Work}

The emergence of neural operators has revolutionized data-driven PDE solving by learning continuous mappings between function spaces, bypassing traditional discretization bottlenecks. Two pioneering architectures—DeepONet~\cite{deeponet} and Fourier Neural Operator (\ac{FNO})~\cite{li2021fourier}—established the foundation for operator learning, inspiring a rich ecosystem of extensions targeting accuracy, efficiency, geometry adaptability, and physical consistency.

\textbf{FNO-based architectures} have primarily evolved along three axes: \emph{spectral efficiency}, \emph{domain flexibility}, and \emph{feature fusion}. Factorized-FNO~\cite{tran2021factorized} introduced separable spectral convolutions and enhanced residual connections, significantly improving convergence and generalization on both regular and scattered grids. To overcome FNO’s inherent limitation to Cartesian domains, GeoFNO~\cite{li2023fourier} proposed learnable domain deformation, enabling high-fidelity simulations on complex geometries with up to 40\% error reduction. Spherical-FNO~\cite{bonev2023spherical} tailored spectral operators to spherical coordinates, achieving unprecedented long-term stability in global climate and atmospheric forecasting. More recently, Conv-FNO~\cite{liu2025ehance} addressed FNO’s weakness in capturing local structures by integrating CNN-based feature extractors, achieving resolution invariance and substantial gains in boundary-sensitive problems. Diffusion-FNO~\cite{liu2025difffno} further pushed the envelope by fusing spectral blocks with diffusion-based refinement, enhancing super-resolution accuracy in turbulent and multiphase flows. Amortized-FNO~\cite{xiao2024amortized} introduced a radical efficiency leap by employing Kolmogorov-Arnold Networks (\ac{KAN})~\cite{liu2024kan} to implicitly encode infinite frequency modes, reducing computational overhead while improving average performance by 31\% across diverse PDE benchmarks.

\textbf{DeepONet-based frameworks} have focused on enhancing physical grounding, temporal dynamics, and geometric conditioning. Physics-informed DeepONet~\cite{wang2021learning} pioneered zero-shot operator learning by embedding PDE residuals directly into the loss, enabling predictions orders of magnitude faster than numerical solvers without requiring paired input-output data. ResUNet DeepONet~\cite{he2023novel} replaced the standard trunk network with a U-Net-style residual architecture, dramatically improving accuracy in predicting elastoplastic stress fields under complex, load-varying geometries. Geom-DeepONet~\cite{geomdeeponet} established a new standard for design-aware modeling by fusing cross-modal geometric descriptors (explicit CAD features + implicit SDFs) and employing \ac{SIREN}~\cite{sitzmann2020implicit} for high-frequency spatial encoding—accelerating parametric simulations by 5–10$\times$. Sequential-DeepONet~\cite{he2024sequential} broke new ground by integrating LSTM/GRU units into the branch network, enabling memory-aware modeling of path-dependent processes such as plasticity and thermal hysteresis, with error reductions of up to 2.5$\times$ compared to static architectures.

Beyond these two main branches, \textbf{hybrid and physics-structured frameworks} are emerging as next-generation paradigms. DeepM\&Mnet~\cite{cai2021deepm} introduced a “plug-and-play” modular framework that composes multiple pretrained DeepONets to assimilate multiphysics data—ideal for systems with coupled phenomena (e.g., fluid-structure interaction). \ac{FVGN}~\cite{li2023finite} bridged traditional finite volume methods with graph neural networks, preserving conservation laws while learning from sparse, unstructured observations. These developments reflect a broader trend: the field is maturing from pure function approximation toward \emph{physics-structured}, \emph{geometry-adaptive}, and \emph{computationally scalable} operator learning—setting the stage for industrial deployment in design optimization, digital twins, and real-time control.

\section{Methodology}

\subsection{Problem Definition}

\begin{figure}[htbp]
    \centering
    \includegraphics[width=0.5\textwidth]{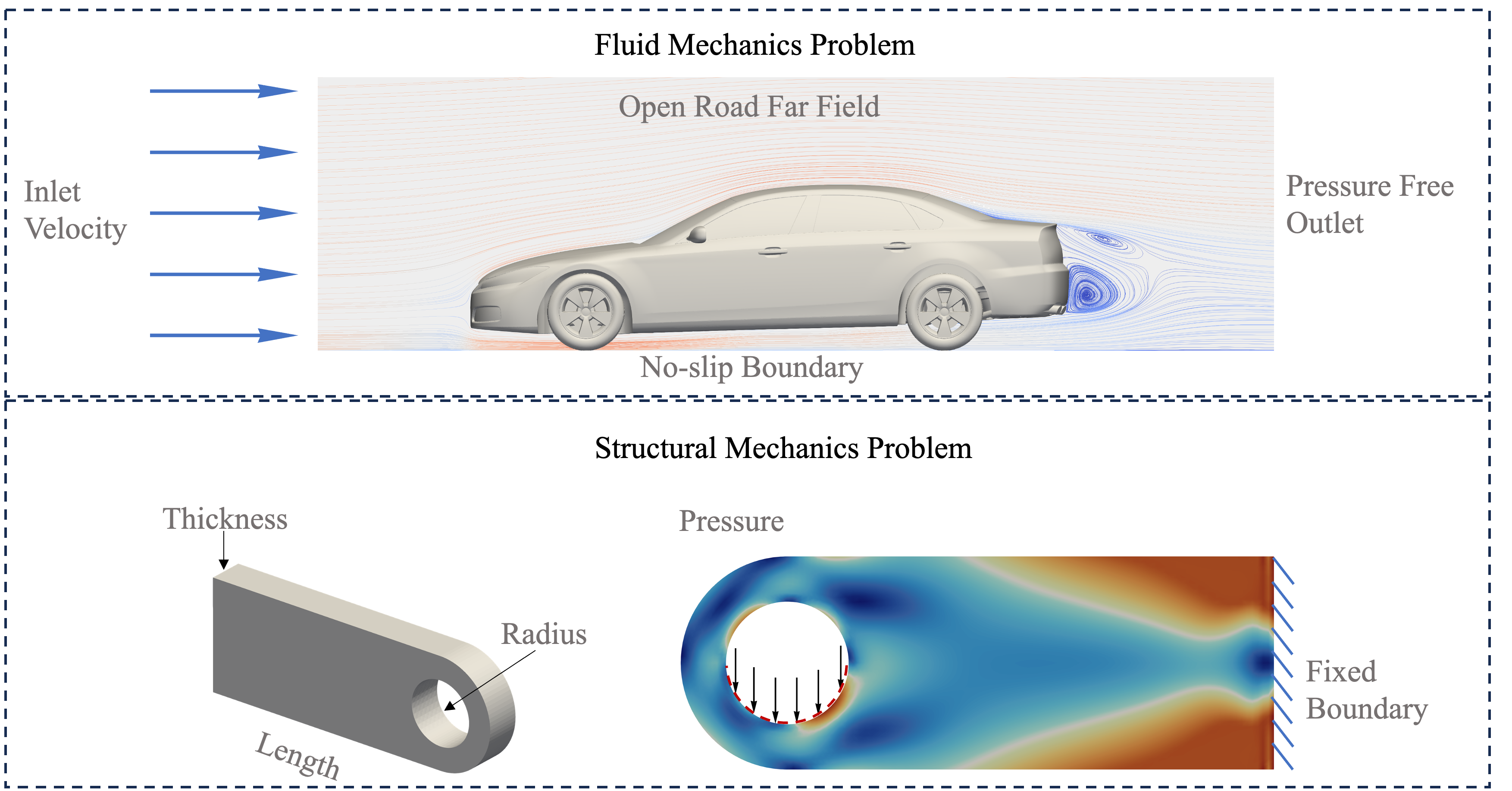}
    \caption{\textbf{Simulation Settings}: (Top) A typical open-road \ac{CFD} simulation for the deformed DrivAer model with three rear-end configurations. Appropriate Dirichlet boundary conditions are applied. (Bottom) At the bottom half of the hole, the cantilever beam is suppressed by a $50\,\text{MPa}$ uniform stress along the $z$-axis, while the flat side is fully constrained.}
    \label{fig:work_load}
\end{figure}

We consider two representative 3D physical simulation problems:  aerodynamics of an automobile and structural mechanics of a cantilever beam with a hole, as illustrated in Fig.~\ref{fig:work_load}. The geometry is represented as a large-scale point cloud, and the physical fields are modeled by solving partial differential equations (PDEs) on this discrete domain.

Let $\Omega \subset \mathbb{R}^3$ be a bounded open set representing the spatial domain, and let $x \in \Omega$ denote a material point. Design parameters such as shape and boundary conditions are represented by $d \in \mathcal{A}$, where $\mathcal{A}$ is a bounded Banach space. The solution fields—stress $\sigma$ and velocity $v$—are defined in a Bochner space $\mathcal{B} \subset L^2(0, T; \mathbf{H}^2(\Omega)) \cap H^1(0, T; \mathbf{L}^2(\Omega))$, ensuring sufficient regularity for physical consistency.

The forward problem is governed by an elliptic PDE system:
\begin{align}
\mathcal{L}(g)(x) &= s(x), \quad x \in \Omega, \\
g(x) &= c, \quad x \in \partial\omega,
\end{align}
where $\mathcal{L}$ is a differential operator, $s(x)$ a source term, and $c$ a boundary condition on $\partial\omega$.

Our goal is to learn an approximation operator $\mathcal{G}$ to the ground-truth solution functional $G^*(u): u(x, d) \to [\sigma, v]$, where $u(x, d)$ encodes both spatial coordinates and design parameters. Training data $\{\hat{u}_i(x_i,d_i), \hat{\sigma}_i, \hat{v}_i\}_{i=1}^N$ are generated from numerical simulations over random geometries and boundary conditions.

We assume the existence of a computable Green's function $G_r(u, y)$ under a Lebesgue measure $\nu(u)$, such that the solution admits an integral representation:
\begin{align}
[\sigma,~v] &= \int_{\Omega} G_r(u,y) f(y) \,dy, \\
[\sigma_{bc},~v_{bc}] &= \int_{\partial\omega} G_r(u,y) f(y) \,dy.
\end{align}

Guided by this formulation, we define a recursive deep neural operator with learnable parameters $\phi$, inspired by the kernel-based architecture in~\cite{li2021fourier}. The approximation $\mathcal{G}(u)$ is constructed via a sequence of integral transformations:
\begin{align}
\mathcal{G}(u) =
\begin{cases}
v_0 = u(x, d), & l = 0, \\[6pt]
\begin{aligned}[b]
v_{l+1} = &\sigma_{\phi}\Bigl( \int_{\omega} \kappa_{\phi}\bigl(u(x,d), y, \\ &a(u), a(y)\bigr) 
\; d\nu(y) + \alpha v_l \Bigr),
\end{aligned}
& l < k.
\end{cases}
\end{align}
where $\kappa_{\phi}$ is a learnable kernel function, $a(\cdot)$ represents spatially varying physical features, and $\nu(y)$ is a measure on subdomain $\omega$.

The objective is to find optimal parameters $\phi \in \mathcal{H}$ that minimize the prediction error:
\begin{equation}
\begin{split}
\phi = &\arg\min_{\phi_{\text{iter}}}
       \sum_{i=1}^{N}
       \mathcal{L}_{\text{loss}}
       \Bigl( G^*\bigl(u_i(x_i,d_i)\bigr) \\ &+ \epsilon(u_i)
              - \mathcal{G}(u_i;\phi_{\text{iter}}) \Bigr).
\end{split}
\end{equation}
where $\epsilon(u_i)$ denotes the numerical discretization error between the true operator $G^*$ and the simulation label $[\sigma, v]$.

Building upon this theoretical foundation, we present \ac{LRQ-Solver}, a transformer-based neural operator that implements $\mathcal{G}$ with enhanced scalability and design awareness. The architecture, illustrated in Fig.~\ref{fig:model_architecture}.

\subsection{Parameter--Conditioned Lagrangian Modelling of Material Points}
\label{sec:PCLM}

In multi-configuration design analysis of engineering systems, geometrically similar local sub-structures may exhibit markedly disparate physical behaviours owing to variations in global scale, proportion, or topology. For example, a curved duct segment can sustain fully-laminar flow in a compact configuration yet precipitate early transition in a larger-scale deployment; an identically filleted joint may concentrate stresses differently under altered aspect ratios. Traditional field models that treat state variables as functions of spatial position alone are inherently blind to such system-level dependencies.

\begin{figure*}[htbp]
    \centering
    \includegraphics[width=1\textwidth]{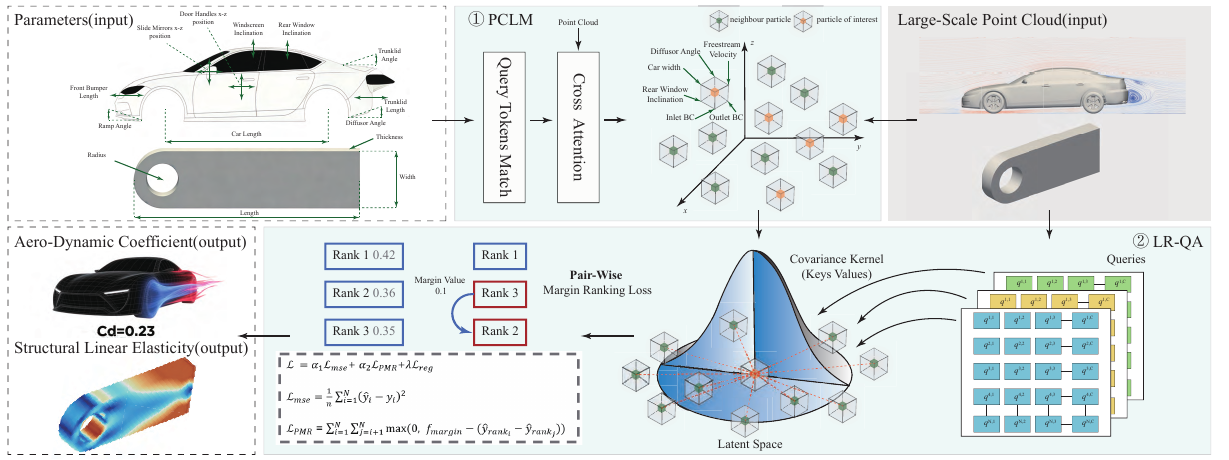}
    \caption{\textbf{Overview of our \ac{LRQ-Solver} Framework.} Parameters are encoded into a latent control field via \textcircled{1} \ac{PCLM}, modulating the physical response. The \textcircled{2} \ac{LR-QA} mechanism computes a covariance kernel to model interactions. A pair-wise ranking loss enforces monotonicity in aerodynamic predictions.}
    \label{fig:model_architecture}
\end{figure*}

To overcome this limitation, we propose a \ac{PCLM} in which the state of every material point is expressed as a joint function of its spatial coordinate and a global shape descriptor. This endows the point with design context awareness along its entire Lagrangian trajectory.

Let $\mathbf{d}\in\mathbb{R}^{m}$ denote the vector of shape parameters characterising a given geometric configuration. \ac{PCE} compresses $\mathbf{d}$ into a low-dimensional semantic context vector

\begin{equation}
    \boldsymbol{\psi} = \mathcal{E}(\mathbf{d}) \in \mathbb{R}^{c},
\end{equation}
where $\mathcal{E}\colon\mathbb{R}^{m}\to\mathbb{R}^{c}$ extracts parameter combinations that dominantly influence global dynamics. Crucially, $\boldsymbol{\psi}$ is not used for geometry generation; instead, it serves as an implicit control field that globally modulates the physical response of every material point during inference.

Consider a Lagrangian material point located at $\mathbf{x}$. Classical treatments express its velocity $\mathbf{v}$, pressure $p$, and temperature $T$ as functions of $\mathbf{x}$ alone. Within the present framework, the solution is generalised to an explicit dependence on the semantic context vector:

\begin{equation}
    \mathbf{u}(\mathbf{x};\boldsymbol{\psi}) = (\mathbf{v}(\mathbf{x};\boldsymbol{\psi}), p(\mathbf{x};\boldsymbol{\psi}), T(\mathbf{x};\boldsymbol{\psi})),
\end{equation}

so that the physical state at a fixed location $\mathbf{x}$ varies systematically with the global configuration encoded in $\boldsymbol{\psi}$.

Under this formulation, the classical conservation laws of mass, momentum, and energy retain their differential forms, but the solution fields are now defined over the extended input space $(\mathbf{x},\boldsymbol{\psi})$:

\begin{gather}
  \nabla\cdot\mathbf{v}(\mathbf{x};\boldsymbol{\psi}) = 0, \\[4pt]
  \rho\,\frac{D\mathbf{v}}{Dt} = -\nabla p(\mathbf{x};\boldsymbol{\psi})
                                  +\mu\nabla^{2}\mathbf{v}(\mathbf{x};\boldsymbol{\psi}), \\[4pt]
  \rho c_{p}\frac{DT}{Dt} = k\nabla^{2}T(\mathbf{x};\boldsymbol{\psi}).
\end{gather}

with material derivative $\mathrm{D}/\mathrm{D}t = \partial/\partial t + \mathbf{v}\cdot\nabla$. While the governing equations preserve physical consistency, their solutions are implicitly shaped by $\boldsymbol{\psi}$ through the boundary conditions, domain geometry, and dimensionless numbers that depend on $\mathbf{d}$.

These system-level effects are injected into the local dynamics of each material point via $\boldsymbol{\psi}$, thereby equipping the Lagrangian particle with design awareness. Consequently, the shape parameters $\mathbf{d}$ transcend their conventional role as mere geometric inputs; they act as an implicit control field that globally modulates the dynamics of every material point through the latent vector $\boldsymbol{\psi}$.

To realise the mapping $\mathcal{E}(\mathbf{d}) = \boldsymbol{\psi}$, we design a structured encoder that extracts semantic context from low-dimensional design parameters and injects it into the high-dimensional physical field predictor. Inspired by BLIP-2~\cite{pmlr-v202-li23q}, which effectively bridges heterogeneous modalities (e.g., vision and language) through learnable query vectors, we adapt this mechanism to the domain of geometric--parametric fusion in physics-informed deep learning.

The PCE operates as a cross-attention bridge between the design parameters and the point-wise field solver. It begins with a set of $N_q = 10$ learnable context queries $\mathbf{Q}_{\text{pce}} \in \mathbb{R}^{N_q \times D_h}$, initialised from a normal distribution, where $D_h = c$ is the dimension of the latent control field $\boldsymbol{\psi}$. Given an input design vector $\mathbf{d} \in \mathbb{R}^{D_{\text{in}}}$, it is first projected into the feature space:
\begin{equation}
    \mathbf{x}_{\text{pce}} = \mathcal{L}_{\text{proj}}(\mathbf{d}) \in \mathbb{R}^{D_h}.
\end{equation}
This projected vector is treated as a singleton key-value input to a multi-head cross-attention module:
\begin{equation}
    \mathbf{Q}'_{\text{pce}} = \text{MultiHeadAttn}(\mathbf{Q}_{\text{pce}}, \mathbf{x}_{\text{pce}}, \mathbf{x}_{\text{pce}}) \in \mathbb{R}^{N_q \times D_h},
\end{equation}
allowing each query to attend to distinct semantic aspects of the design. The output is normalised and passed through a residual feed-forward network:
\begin{equation}
    \mathbf{Q}''_{\text{pce}} = \text{LayerNorm}(\mathbf{Q}'_{\text{pce}} + \text{FFN}(\mathbf{Q}'_{\text{pce}})),
\end{equation}
\begin{equation}
    \boldsymbol{\psi} = \frac{1}{N_q} \sum_{i=1}^{N_q} \mathbf{Q}''_{\text{pce},i} \in \mathbb{R}^{D_h},
\end{equation}
which realises the desired mapping $\mathcal{E}(\mathbf{d}) = \boldsymbol{\psi}$.

This latent vector $\boldsymbol{\psi}$ is then broadcast across all material points in the domain to form a position-invariant context field, which is concatenated with their spatial coordinates and fed into the downstream field predictor.

\subsection{Physics-Integrated Modeling via Pseudo-Physics Fields}
\label{sec:method_physics}

\textbf{Assumption 1 (Low-rank Structure of Physical Fields):} The discrete representation of physical fields on large-scale point clouds exhibits low-rank structure, meaning there exists a subspace of dimension $r \ll N$ that effectively captures the main features of the physical field.

Given a point cloud $\mathcal{P} = \{\mathbf{x}_i\}_{i=1}^N$, input features $\mathbf{X}^{(0)} \in \mathbb{R}^{N \times D}$ are processed through $L$ network layers. At layer $\ell$, queries, keys, and values are computed as:
\begin{align}
    \mathbf{Q}^{(\ell)} &= \mathcal{L}_Q(\mathbf{X}^{(\ell-1)}), \\
    \mathbf{K}^{(\ell)} &= \mathcal{L}_K(\mathbf{X}^{(\ell-1)}), \\
    \mathbf{V}^{(\ell)} &= \mathcal{L}_V(\mathbf{X}^{(\ell-1)}),
\end{align}
where $\mathbf{Q}^{(\ell)}, \mathbf{K}^{(\ell)}, \mathbf{V}^{(\ell)} \in \mathbb{R}^{N \times C}$, with $C$ being the feature dimension and $C \ll N$.

Spatial relationships are encoded via Rotary Position Embedding:
\begin{equation}
    \mathbf{Q}^{(\ell)}, \mathbf{K}^{(\ell)} = \text{RoPE}(\mathbf{Q}^{(\ell)}, \mathbf{K}^{(\ell)}).
\end{equation}

Standard self-attention computes $\mathbf{Q}^{(\ell)}(\mathbf{K}^{(\ell)})^\top \in \mathbb{R}^{N \times N}$ with $O(N^2)$ complexity. Our method computes the covariance matrices:
\begin{align}
    \mathbf{C}_k^{(\ell)} &= (\mathbf{K}^{(\ell)})^\top \mathbf{K}^{(\ell)} \in \mathbb{R}^{C \times C}, \\
    \mathbf{C}_v^{(\ell)} &= (\mathbf{V}^{(\ell)})^\top \mathbf{V}^{(\ell)} \in \mathbb{R}^{C \times C},
\end{align}
and then calculates the attention output:
\begin{equation}
    \mathbf{Z}^{(\ell)} = \mathbf{Q}^{(\ell)} \mathbf{C}_k^{(\ell)} \mathbf{C}_v^{(\ell)} \mathbf{V}^{(\ell)} \in \mathbb{R}^{N \times C}.
\end{equation}

We provide theoretical justification for the covariance-based attention mechanism under the assumption of low-rank structure in physical fields. The key insight is that when the key matrix $\mathbf{K}^{(\ell)} \in \mathbb{R}^{N \times C}$ has low effective rank, its second-order statistics—captured by the covariance matrix $\mathbf{C}_k^{(\ell)} = (\mathbf{K}^{(\ell)})^\top \mathbf{K}^{(\ell)}$—sufficiently encode the dominant interaction modes, enabling accurate approximation of standard attention with significantly reduced complexity.

\begin{theorem}[Approximation Guarantee of Covariance Attention]
Assume the physical field satisfies a low-rank structure, i.e., $\mathbf{K}^{(\ell)}$ has rank $r \ll C$. Let $\mathbf{K}^{(\ell)} = \mathbf{U}\mathbf{\Sigma}\mathbf{V}^\top$ be the singular value decomposition (SVD) of $\mathbf{K}^{(\ell)}$, where $\mathbf{U} \in \mathbb{R}^{N \times r}$, $\mathbf{\Sigma} \in \mathbb{R}^{r \times r}$, and $\mathbf{V} \in \mathbb{R}^{C \times r}$. Then, the covariance attention output $\mathbf{Z}^{(\ell)}$ and the standard attention output $\mathbf{Z}_{\text{std}}^{(\ell)} = \mathbf{Q}^{(\ell)}(\mathbf{K}^{(\ell)})^\top \mathbf{V}^{(\ell)}$ satisfy:
\begin{equation}
\begin{split}
\|\mathbf{Z}^{(\ell)} - \mathbf{Z}_{\text{std}}^{(\ell)}\|_F
&\leq \|\mathbf{Q}^{(\ell)}\|_F \cdot \|\mathbf{V}^{(\ell)}\|_F \\
&\quad \cdot \|\mathbf{K}^{(\ell)} - \mathbf{K}^{(\ell)}\mathbf{K}^{(\ell)\top}\mathbf{K}^{(\ell)}\|_F,
\end{split}
\end{equation}
where $\|\cdot\|_F$ denotes the Frobenius norm.
\end{theorem}

\begin{proof}
The standard self-attention computes:
\begin{equation}
\mathbf{Z}_{\text{std}}^{(\ell)} = \mathbf{Q}^{(\ell)} (\mathbf{K}^{(\ell)})^\top \mathbf{V}^{(\ell)}.
\end{equation}

Our covariance-based attention computes:
\begin{equation}
\mathbf{Z}^{(\ell)} = \mathbf{Q}^{(\ell)} \left[(\mathbf{K}^{(\ell)})^\top \mathbf{K}^{(\ell)}\right] \left[(\mathbf{V}^{(\ell)})^\top \mathbf{V}^{(\ell)}\right] \mathbf{V}^{(\ell)}.
\end{equation}

For theoretical analysis, we consider the symmetric case where $\mathbf{V}^{(\ell)} = \mathbf{K}^{(\ell)}$, which is common in many physical modeling scenarios. The difference becomes:
\begin{align}
\mathbf{Z}^{(\ell)} - \mathbf{Z}_{\text{std}}^{(\ell)} 
&= \mathbf{Q}^{(\ell)} \left[
     (\mathbf{K}^{(\ell)})^\top \mathbf{K}^{(\ell)} (\mathbf{K}^{(\ell)})^\top \mathbf{K}^{(\ell)} \mathbf{K}^{(\ell)} \right. \notag \\
&\quad \left. - (\mathbf{K}^{(\ell)})^\top \mathbf{K}^{(\ell)}
   \right] \notag \\
&= \mathbf{Q}^{(\ell)} \left[
     (\mathbf{K}^{(\ell)})^\top \mathbf{K}^{(\ell)} - \mathbf{I}
   \right] (\mathbf{K}^{(\ell)})^\top \mathbf{K}^{(\ell)},
\end{align}
where $\mathbf{I}$ is the identity matrix.

Applying the submultiplicativity of the Frobenius norm:
\begin{align}
\|\mathbf{Z}^{(\ell)} - \mathbf{Z}_{\text{std}}^{(\ell)}\|_F 
&\leq \|\mathbf{Q}^{(\ell)}\|_F \cdot
        \|(\mathbf{K}^{(\ell)})^\top \mathbf{K}^{(\ell)} - \mathbf{I}\|_F \cdot
        \|(\mathbf{K}^{(\ell)})^\top \mathbf{K}^{(\ell)}\|_F \notag \\
&\leq \|\mathbf{Q}^{(\ell)}\|_F \cdot
        \|\mathbf{K}^{(\ell)}\|_F \cdot
        \|\mathbf{K}^{(\ell)} - \mathbf{K}^{(\ell)}\mathbf{K}^{(\ell)\top}\mathbf{K}^{(\ell)}\|_F.
\end{align}

Under the low-rank assumption, $\mathbf{K}^{(\ell)}$ admits a compact SVD representation where higher-order singular values decay rapidly. Consequently, the residual term $\|\mathbf{K}^{(\ell)} - \mathbf{K}^{(\ell)}\mathbf{K}^{(\ell)\top}\mathbf{K}^{(\ell)}\|_F$ becomes negligible, as it primarily captures noise or fine-grained fluctuations beyond the dominant coherent structures.

In physical systems such as fluid dynamics or structural mechanics, these dominant modes correspond to large-scale vortices, stress concentrations, or deformation patterns—precisely the features that govern system behavior. Therefore, the covariance attention preserves the physically meaningful interactions while discarding computationally expensive, low-energy noise modes.

This approximation reduces the attention complexity from $O(N^2)$ to $O(NC^2 + C^3)$, enabling scalable simulation of up to 2 million points on a single GPU without sacrificing physical fidelity.
\end{proof}

Feature representations are updated via residual connection:
\begin{equation}
    \mathbf{X}^{(\ell)} = \mathbf{X}^{(\ell-1)} + \mathcal{L}_{\text{out}}^{(\ell)}(\mathbf{Z}^{(\ell)}).
\end{equation}

The network outputs a \emph{pseudo-physics field} $\hat{\mathbf{u}}(\mathbf{x}_i) = (\hat{\mathbf{v}}_i, \hat{p}_i, \hat{T}_i, \hat{\boldsymbol{\sigma}}_i)$, with physical consistency enforced through conservation law residuals:
\begin{equation}
\begin{split}
\mathcal{L}_{\text{phys}} = \frac{1}{N} \sum_{i=1}^{N} \Bigl[\,
    &\bigl\| \nabla \cdot \hat{\mathbf{v}}_i \bigr\|^{2} \\
    &+ \bigl\| \rho (\hat{\mathbf{v}}_i \cdot \nabla) \hat{\mathbf{v}}_i
             + \nabla \hat{p}_i
             - \nabla \cdot \boldsymbol{\hat{\tau}}_i \bigr\|^{2} \\
    &+ \bigl\| \rho c_p (\hat{\mathbf{v}}_i \cdot \nabla \hat{T}_i)
             - \nabla \cdot (k \nabla \hat{T}_i) \bigr\|^{2}
\Bigr].
\end{split}
\end{equation}
System-level responses are computed via control volume integrals:
\begin{align}
    \hat{\mathbf{F}} &= \sum_{i \in \mathcal{S}} \left[ \rho (\hat{\mathbf{v}}_i \cdot \mathbf{n}_i) \hat{\mathbf{v}}_i - \hat{p}_i \mathbf{n}_i + \boldsymbol{\hat{\tau}}_i \cdot \mathbf{n}_i \right] \Delta A_i, \\
    \hat{Q} &= \sum_{i \in \mathcal{S}} (-k \nabla \hat{T}_i \cdot \mathbf{n}_i) \Delta A_i, \\
    \hat{U} &= \sum_{i \in \mathcal{B}} \frac{1}{2} \boldsymbol{\hat{\sigma}}_i : \boldsymbol{\hat{\varepsilon}}_i \Delta V_i,
\end{align}
where $\mathcal{S}$ and $\mathcal{B}$ denote control surface points and body points, respectively.

The total loss function combines multiple objectives:
\begin{equation}
    \mathcal{L} = \alpha_1 \|\hat{y} - y^{\text{true}}\|^2 
                 + \alpha_2 \mathcal{L}_{\text{phys}} 
                 + \alpha_3 \mathcal{L}_{\text{rank}} 
                 + \lambda \|\theta\|^2,
\end{equation}
where the ranking loss is defined as:

\begin{equation}
\begin{split}
\mathcal{L}_{\text{rank}}
&= \sum_{i<j} \max\!\bigl(0,\; m - (\hat{y}_i - \hat{y}_j)\, s_{ij}\bigr), \\
&\qquad \quad
s_{ij} = \operatorname{sign}(y_i^{\text{true}} - y_j^{\text{true}}).
\end{split}
\end{equation}

\section{Experiment}

\textbf{Implementations}~
Our experiments are conducted on 4 NVIDIA A100 40GB PCIe GPUs using the \textbf{PaddlePaddle} framework. We employ PaddlePaddle's distributed data parallel (DDP) training to scale across all devices. The model is optimized with AdamW using an initial learning rate of $1 \times 10^{-4}$, decayed by a factor of 0.1 after 50 epochs. The batch size is set to 4 per GPU during training. 

\textbf{Metrics}~
We adopt a comprehensive set of metrics to evaluate both the accuracy and computational efficiency of neural PDE solvers. For accuracy assessment, we use four widely adopted error measures: Mean Squared Error (MSE), Mean Absolute Error (MAE), Maximum Absolute Error (Max AE), and Mean Relative Error (MRE). The MSE measures the average squared deviation between predicted and ground-truth values and is defined as $\text{MSE} = \frac{1}{n}\sum_{i=1}^n (y_i - \hat{y}_i)^2$, making it sensitive to large errors. The MAE computes the average absolute difference, $\text{MAE} = \frac{1}{n}\sum_{i=1}^n |y_i - \hat{y}_i|$, providing a robust evaluation less influenced by outliers. The Max AE captures the worst-case prediction error, $\text{Max AE} = \max_i |y_i - \hat{y}_i|$, which is critical for safety-critical engineering applications where peak deviations must be minimized. The MRE normalizes the error by the magnitude of the true values, $\text{MRE} = \frac{1}{n}\sum_{i=1}^n \frac{|y_i - \hat{y}_i|}{|y_i|} \times 100\%$, enabling fair comparison across datasets with varying scales and units.

For computational efficiency, we report Training Time, Inference Time, and FLOPs. Training Time is measured in hours and reflects the total wall-clock time required to complete model training on the given hardware. Inference Time denotes the latency (in seconds) of a single forward pass, which is crucial for real-time or iterative design workflows. All models are trained using single-precision (FP32) arithmetic to ensure a fair comparison in both accuracy and computational cost.

\subsection{Main Results}
We demonstrate that \ac{LRQ-Solver} achieves state-of-the-art performance in large-scale 3D industrial physics simulation, outperforming existing methods in both accuracy and computational efficiency. We evaluate our model on the \textit{DrivAerNet++} dataset, a comprehensive benchmark for vehicle aerodynamics, and present a detailed analysis of its predictive capability and scalability.

\subsubsection{3D complex Aerodynamics turbulence problem}

We evaluate our approach on the \textit{DrivAerNet++} dataset, a large-scale benchmark for industrial 3D vehicle aerodynamics simulation. The dataset contains 8,000 high-resolution 3D vehicle models with 23 deformable geometric control parameters and CFD-simulated aerodynamic labels, including the total drag coefficient $C_d$. We uniformly sample 100k points per model while preserving geometric fidelity and surface detail distribution. The data is split into 70\% training, 15\% validation, and 15\% testing.

The underlying physics is governed by the incompressible Navier-Stokes equations in non-dimensional form:
\begin{align}
\frac{\partial \mathbf{u}}{\partial t} + \mathbf{u} \cdot \nabla \mathbf{u}
  &= -\nabla p + \frac{1}{\mathrm{Re}} \nabla^2 \mathbf{u}, \\
\nabla \cdot \mathbf{u} &= 0,
\end{align}
where $\mathbf{u}$ is the velocity field, $p$ the pressure, and $\mathrm{Re}$ the Reynolds number. The total drag coefficient $C_d$ is computed via surface integration of the pressure and viscous stress fields over the wetted geometry:
\begin{equation}
C_d = \frac{1}{\frac{1}{2}\rho U^2 A_{\text{ref}}} \int_{\mathcal{S}} \left( p \mathbf{n} + \boldsymbol{\tau} \cdot \mathbf{n} \right) \cdot \mathbf{e}_x  \,dA,
\end{equation}
where $\mathcal{S}$ is the vehicle surface, $\mathbf{n}$ the outward normal, $\boldsymbol{\tau}$ the viscous stress tensor, and $\mathbf{e}_x$ the flow direction.

\paragraph{Accuracy Comparison}
We compare \ac{LRQ-Solver} against a comprehensive set of baselines: GCNN~\cite{kipf2016semi}, RegDGCNN~\cite{drivaernetplus}, PointNet~\cite{drivaernetplus}, Transolver~\cite{wu2024transolver}, Transolver++~\cite{luo2025transolverplus}, DAT~\cite{he2025drivaer}, and TripNet~\cite{chen2025tripnet}. As shown in Table~\ref{tab:performance_drivaernet}, \ac{LRQ-Solver} achieves state-of-the-art accuracy across all metrics, reducing the MSE by 38.9\% compared to the previous best (TripNet). The model attains an MRE of \textbf{2.25\%}, which is remarkably close to the theoretical precision limit of 2.18\%—the observed discrepancy between high-fidelity CFD and wind tunnel experiments. This suggests that \ac{LRQ-Solver} has nearly reached the effective accuracy ceiling of the dataset.

\begin{table}[H]
\centering
\caption{\textbf{Accuracy Comparison on \textit{DrivAerNet++}}. Bold values indicate the best results; underlined values denote the second-best. Our model outperforms the current best network on the \textbf{DrivAerNet++ Leaderboard} across 1,163 industry-standard car designs, approaching to the theoretical precision limit (MRE = $2.18\%$ between wind tunnel experiments and \textit{DrivAerNet++}).}
\label{tab:performance_drivaernet}
\resizebox{0.5\textwidth}{!}{
\begin{tabular}{l c c c c}
\toprule
Model & $\text{MSE} \times 10^{-5}$ & $\text{MAE} \times 10^{-3}$ & $\text{Max AE} \times 10^{-2}$ & MRE \\
\midrule
GCNN & 17.10 & 10.43 & 15.03 & -- \\
RegDGCNN & 14.20 & 9.31 & 12.79 & -- \\
PointNet & 14.90 & 9.60 & 12.45 & -- \\
Transolver & 60.30 & 20.31 & 65.61 & -- \\
Transolver++ & 46.10 & 17.69 & 57.95 & -- \\
DAT & 11.80 & 9.11 & 11.20 & -- \\
TripNet & \underline{9.10} & \underline{7.17} & \underline{7.70} & -- \\
\midrule
\ac{LRQ-Solver} (Ours) & \textbf{5.56} & \textbf{5.90} & \textbf{3.22} & \textbf{2.25\%} \\
\bottomrule
\end{tabular}}
\end{table}

Fig~\ref{fig:drivaerpp_train} shows stable training and validation loss convergence. Fig~\ref{fig:MSE_car_type} further illustrates consistent accuracy across different vehicle configurations, with particularly strong gains in complex rear-end designs prone to flow separation.

\begin{figure}[htbp]
    \centering
    \includegraphics[width=0.4\textwidth]{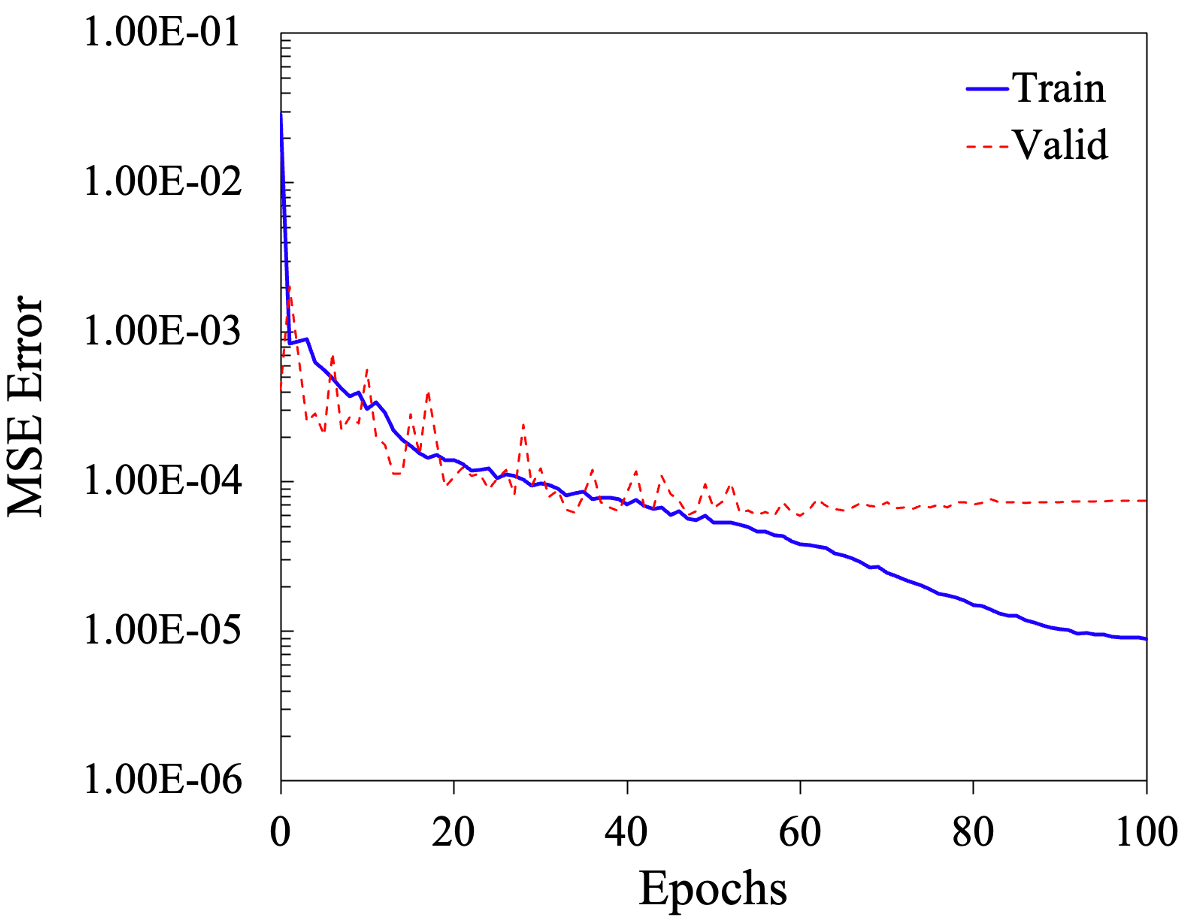}
    \caption{Train and valid loss of \textit{DrivAerNet++}.}
    \label{fig:drivaerpp_train}
\end{figure}

\begin{figure}[htbp]
    \centering
    \includegraphics[width=0.4\textwidth]{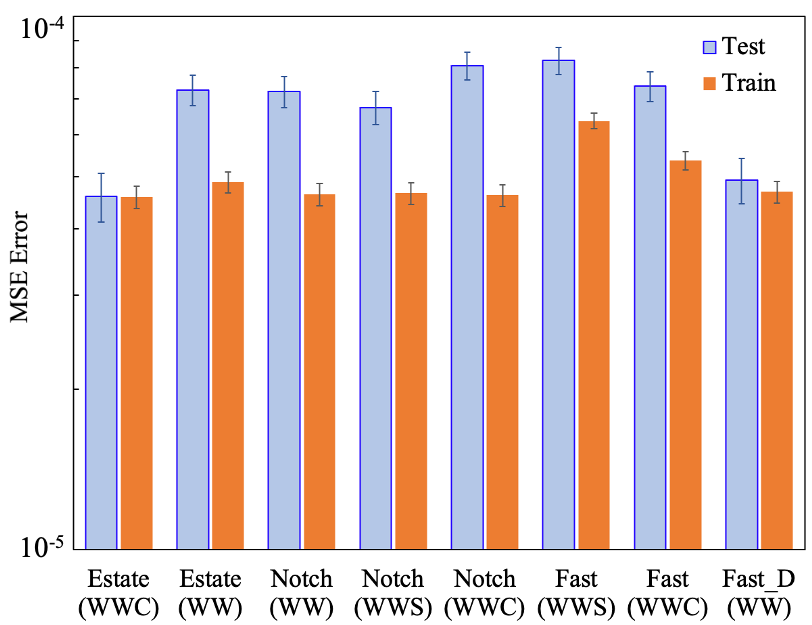}
    \caption{MSE error of different Vehicle rear type and wheels format in \textit{DrivAerNet++}.}
    \label{fig:MSE_car_type}
\end{figure}

\paragraph{Efficiency and Scalability}
Beyond accuracy, \ac{LRQ-Solver} achieves exceptional computational efficiency. As summarized in Table~\ref{tab:efficiency_drivaerpp}, the model completes training in just \textbf{7.2 hours}, over 6× faster than Transolver++ (46.1h) and 2× faster than DAT (14.3h), despite using only 16 batch size (half of most baselines) and 4× A100 40G GPUs—less specialized hardware than H20 or H100 used by competitors.

More significantly, inference latency is reduced to \textbf{5 milliseconds}, representing a 126× speedup over Transolver++ (0.63s) and over 140× improvement compared to DAT (0.73s). This enables real-time performance prediction and seamless integration into iterative design workflows.

\ac{LRQ-Solver} delivers both high fidelity and real-time efficiency, making it uniquely suitable for industrial-scale deployment. \ac{LRQ-Solver} not only surpasses all existing methods in accuracy but also achieves unprecedented efficiency, positioning it as a scalable, high-fidelity solution for real-world engineering design and optimization.

\begin{table*}[htbp]
\centering
\caption{\textbf{Efficiency Comparison on \textit{DrivAerNet++}}. 
\ac{LRQ-Solver} achieves the fastest training and inference times despite using less specialized hardware and a smaller batch size. 
Single-precision (FP32) performance is provided for fair comparison, as all models are trained in FP32. Despite being evaluated on a single NVIDIA A100 GPU (weaker than the multi-GPU or H100 and H20 setups used in prior work) and with a smaller batch size, \ac{LRQ-Solver} achieves the fastest training and inference times—demonstrating superior algorithmic efficiency rather than hardware advantage. }
\label{tab:efficiency_drivaerpp}
\resizebox{\textwidth}{!}{
\begin{tabular}{l c c c c c c}
\toprule
Model & Train Time (h) & Inference Time (s) & Epochs & Batch Size & Hardware & FP32 Performance \\
\midrule
GCNN & 49.0 & 50.80 & 100 & 32 & 4 H20 96G & 44 TFLOPS \\
RegDGCNN & 12.6 & 0.85 & 100 & 32 & 4 H20 96G & 44 TFLOPS \\
PointNet & 4.7 & 0.66 & 100 & 32 & 4 H20 96G & 44 TFLOPS \\
Transolver & 45.7 & 0.66 & 100 & 32 & 4 H20 96G & 44 TFLOPS \\
Transolver++ & 46.1 & 0.63 & 100 & 32 & 4 H20 96G & 44 TFLOPS \\
DAT & 14.3 & 0.73 & 100 & 32 & 4 H20 96G & 44 TFLOPS \\
TripNet & -- & -- & 200 & 32 & 4 H100 80G & 48 TFLOPS \\
\midrule
\ac{LRQ-Solver} (Ours) & \textbf{7.2} & \textbf{0.005} & 100 & 16 & 4 A100 PCIe 40G & 19.5 TFLOPS \\
\bottomrule
\end{tabular}}
\end{table*}

\subsubsection{3D Structural Linear Elasticity Problem}
\label{sec:3d_beam}

A linear elasticity system is governed by the following equations. The equilibrium equation describes the balance of stress and body forces:
\begin{equation}
\sigma_{ij,j} + F_i = 0,
\end{equation}
where $\sigma_{ij}$ is the stress tensor and $F_i$ represents body forces. The strain $\epsilon_{ij}$ is related to the displacement field $u_i$ through the kinematic relation:
\begin{equation}
\epsilon_{ij} = \frac{1}{2} \left( \frac{\partial u_i}{\partial x_j} + \frac{\partial u_j}{\partial x_i} \right).
\end{equation}
The constitutive behavior of the material follows Hooke's law for isotropic elasticity:
\begin{equation}
\sigma_{ij} = \frac{E}{1+\nu} \left( \epsilon_{ij} + \frac{\nu}{1-2\nu} \epsilon_{kk} \delta_{ij} \right),
\end{equation}
where $E$ is Young's modulus, $\nu$ is Poisson's ratio, and $\delta_{ij}$ is the Kronecker delta.

We focus on predicting the equivalent Von Mises stress, a critical indicator for assessing whether the applied load reaches the material's yield strength:
\begin{equation}
\sigma_{vM} = \sqrt{\frac{3}{2} \mathbf{S} : \mathbf{S}},
\end{equation}
where $\mathbf{S}$ is the deviatoric stress tensor.

We evaluate our method on the \textit{3D Beam} dataset, which contains finite element simulation results of elastic beams under various geometries, variable pressure loads, and fixed boundary conditions. The dataset includes point clouds at multiple resolutions (250–25k points), enabling resolution-agnostic evaluation. The label is the nodal von Mises stress field $\sigma_{vM}$, used to assess both accuracy and generalization. The right flat side of the beam is fully fixed (all degrees of freedom constrained), and a uniform pressure load is applied over the bottom half of the central hole. The material properties are: Young's modulus $E = 200~\text{GPa}$, Poisson's ratio $\nu = 0.3$, yield strength $380~\text{MPa}$, and hardening modulus $571.4~\text{MPa}$. The dataset consists of 3,000 unique configurations with varying geometric and loading inputs, partitioned into $75\%$ training, $5\%$ validation, and $20\%$ testing sets.

\paragraph{Accuracy Comparison}
We compare \ac{LRQ-Solver} against established baselines: DeepONet~\cite{deeponet}, Geom-DeepONet~\cite{geomdeeponet}, Transolver~\cite{wu2024transolver}, and RegDGCNN~\cite{drivaernetplus}. As shown in Table~\ref{tab:accuracy_3dbeam}, \ac{LRQ-Solver} achieves a MAE of $1.66~\text{MPa}$ on the full geometry, representing a $28.8\%$ improvement over the previous best result (Geom-DeepONet, $2.33~\text{MPa}$). The model also achieves the lowest error on the commonly used 5k-node subset, demonstrating superior fidelity in both localized and global stress prediction. Notably, RegDGCNN fails to run on the full geometry due to out-of-memory (OOM) on a single A100 GPU (40GB), highlighting its limited scalability. In contrast, \ac{LRQ-Solver} successfully processes the full point cloud with minimal error, indicating strong generalization and memory efficiency.

\begin{table}[H]
\centering
\caption{\textbf{Accuracy Comparison on the \textit{3D Beam} Dataset.} 
Bold values denote the best results; \underline{underlined} values denote the second-best. 
MAE is evaluated at two levels: \textbf{MAE$_{\text{subset}}$} on a 5k-node subset commonly used for benchmarking, and \textbf{MAE$_{\text{all}}$} over the full 3D volume for comprehensive assessment. 
OOM indicates out-of-memory on a single A100 GPU (40GB), meaning the model cannot process the full geometry.}
\label{tab:accuracy_3dbeam}
\resizebox{0.5\textwidth}{!}{
\begin{tabular}{l c c}
\toprule
Model & MAE$_{\text{subset}}$ (5k nodes) & MAE$_{\text{all}}$ (Full Volume) \\
\midrule
DeepONet & 7.14 & 7.16 \\
RegDGCNN & 34.75 & OOM \\
Transolver & 37.95 & 37.69 \\
Geom-DeepONet & \underline{2.33} & \underline{2.32} \\
\midrule
\ac{LRQ-Solver} (Ours) & \textbf{1.66} & \textbf{1.66} \\
\bottomrule
\end{tabular}
}
\end{table}

\paragraph{Efficiency and Scalability}
In terms of computational efficiency, \ac{LRQ-Solver} achieves unprecedented training speed, completing training in just \textbf{0.76 hours}—over 60× faster than DeepONet (27.5h) and 61× faster than Geom-DeepONet (46.9h). Despite the significant speedup, the model maintains competitive inference latency of \textbf{3.5 ms}, on par with DeepONet (2.4 ms) and Geom-DeepONet (2.7 ms), as shown in Table~\ref{tab:efficiency_3dbeam}. The efficiency gains stem from the covariance-based low-rank attention mechanism, which reduces computational complexity and memory footprint, enabling full-geometry processing within single-GPU memory limits. With only 2,000 epochs (vs. 150,000 for DeepONet variants), \ac{LRQ-Solver} achieves rapid convergence, making it highly suitable for iterative engineering design and real-time simulation workflows.

\begin{table}[H]
\centering
\caption{\textbf{Efficiency Comparison on the \textit{3D Beam} Dataset.} 
All models are evaluated with batch size 16 on a single A100 GPU (40GB). Inference time is measured per sample (ms). Training time is total wall-clock time in hours.}
\label{tab:efficiency_3dbeam}
\resizebox{0.5\textwidth}{!}{
\begin{tabular}{l c c c}
\toprule
Model & Training Time (h) & Inference Time (ms) & Epochs \\
\midrule
DeepONet & 27.5 & \textbf{2.4} & 150,000   \\
RegDGCNN & 18.4 & 16,836.6 & 2,000   \\
Transolver & 13.0 & 3,348.3 & 2,000  \\
Geom-DeepONet & 46.9 & 2.7 & 150,000  \\
\midrule
\ac{LRQ-Solver} (Ours) & \textbf{0.76} & 3.5 & 2,000  \\
\bottomrule
\end{tabular}
}
\end{table}

\subsection{Ablation Study}
\label{sec:ablation}

We conduct ablation studies on the \textit{DrivAerNet++} and \textit{3D Beam} datasets to evaluate the contribution of each component in \ac{LRQ-Solver}. The results, shown in Table~\ref{tab:ablation_study_drivaerpp} and Table~\ref{tab:ablation_study_beam}, are based on a baseline \ac{MLP} without \ac{LR-QA} or \ac{PCLM}, with components added incrementally.

Adding \ac{LR-QA} alone significantly improves accuracy on \textit{DrivAerNet++}, reducing MSE from 36.81 to 8.98 ($\times 10^5$) and MRE from 5.63\% to 2.82\%. On \textit{3D Beam}, \ac{LR-QA} reduces MAE from 3.01 to 1.99 MPa on both the 5k-node subset and full geometry. This demonstrates that the covariance-based low-rank attention mechanism effectively captures long-range physical interactions, leading to substantial gains in prediction fidelity, particularly in regions with complex flow structures or stress gradients.

In contrast, adding \ac{PCLM} alone yields a more moderate improvement on \textit{3D Beam} (MAE reduced to 2.37 MPa) but achieves the lowest MRE (2.22\%) on \textit{DrivAerNet++}, outperforming \ac{LR-QA} in relative error. This suggests that \ac{PCLM} is particularly effective at modeling global design dependencies, such as rear-end configurations and thickness variations, where system-level parameters strongly influence local physical behavior.

When both \ac{LR-QA} and \ac{PCLM} are combined in \ac{LRQ-Solver}, the model achieves the best performance on both datasets: MSE drops to 5.56 ($\times 10^{-5}$) on \textit{DrivAerNet++} and MAE reaches 1.66 MPa on \textit{3D Beam}, outperforming all ablated variants. The full model not only improves absolute accuracy but also maintains consistent performance across different evaluation granularities (subset vs. full volume), indicating robust generalization.

These results confirm that \ac{LR-QA} and \ac{PCLM} play complementary roles: \ac{LR-QA} enhances spatial coherence modeling, while \ac{PCLM} enables configuration-aware prediction. Neither component alone is sufficient to achieve optimal performance—only their integration enables \ac{LRQ-Solver} to simultaneously capture long-range physical interactions and global design effects, achieving state-of-the-art accuracy in multi-configuration industrial simulations.

\begin{table}[H]
\centering
\caption{Ablation study on the \textit{DrivAerNet++} dataset. Baseline is a \ac{MLP} without \ac{LR-QA} or \ac{PCLM}. `w/' denotes `with', indicating the addition of the corresponding component to the baseline model.}
\label{tab:ablation_study_drivaerpp}
\resizebox{0.5\textwidth}{!}{
\begin{tabular}{l c c c c}
\toprule
Model & $\text{MSE} \times 10^{-5}$ & $\text{MAE} \times 10^{-3}$ & $\text{Max AE} \times 10^{-2}$ & MRE \\
\midrule
Baseline (\ac{MLP})  & 36.81 & 14.64 & 6.78 & 5.63\% \\
\quad w/ LR-QA & 8.98 & 7.41 & 4.09 & 2.82\% \\
\quad w/ PCLM & 5.69 & 5.81 & 3.56 & \textbf{2.22\%} \\
\midrule
\ac{LRQ-Solver} (Ours) & \textbf{5.56} & \textbf{5.90} & \textbf{3.22} & 2.25\% \\
\bottomrule
\end{tabular}}
\end{table}

\begin{table}[h]
\centering
\caption{Ablation study on the \textit{3D Beam} dataset. Baseline is a \ac{MLP} without \ac{LR-QA} or \ac{PCLM}. `w/' denotes `with', indicating the addition of the corresponding component to the baseline model. 
\textbf{MAE$_{\text{subset}}$} is evaluated on the 5k-node subset; \textbf{MAE$_{\text{all}}$} is computed over the full 3D volume.}
\label{tab:ablation_study_beam}
\resizebox{0.5\textwidth}{!}{
\begin{tabular}{l c c}
\toprule
\textbf{Model} & \textbf{MAE$_{\text{subset}}$ (5k nodes)} & \textbf{MAE$_{\text{all}}$ (Full Volume)} \\
\midrule
Baseline (\ac{MLP}) & 3.01 & 3.07 \\
\quad w/ LR-QA  & 1.99 & 1.99 \\
\quad w/ PCLM & 2.37 & 2.38 \\
\midrule
\ac{LRQ-Solver} (Ours) & \textbf{1.66} & \textbf{1.66} \\
\bottomrule
\end{tabular}
}
\end{table}

\subsection{Discretization Invariance Analysis}
\label{sec:discretization_invariance}

\begin{table*}[htbp]
\centering
\caption{Performance of \ac{LRQ-Solver} on the \textit{DrivAerNet++} dataset across different point cloud sizes.}
\label{tab:point_num_exploration_drivaerpp}
\resizebox{\textwidth}{!}{%
\begin{tabular}{lcccccccc}
\toprule
\textbf{Point Numbers} & $\textbf{MSE} \times 10^{-5}$ & $\textbf{MAE} \times 10^{-3}$ & $\textbf{Max AE} \times 10^{-2}$ & \textbf{MRE} & \textbf{Training Time} & \textbf{Inference Time} & \textbf{Memory} & \textbf{FLOPs} \\
\midrule
1024    & 6.40 & 6.20 & 3.14 & 2.37\% & 0.3\,h & 0.005\,s & 0.03\,GB & 2.92\,G \\
4096    & 6.06 & 6.05 & 3.22 & 2.31\% & 0.4\,h & 0.005\,s & 0.08\,GB & 11.69\,G \\
8192    & 5.89 & 5.99 & 3.42 & 2.28\% & 0.6\,h & 0.005\,s & 0.14\,GB & 23.37\,G \\
16384   & 5.63 & \textbf{5.87} & \textbf{2.83} & \textbf{2.23}\% & 1.06\,h & 0.008\,s & 0.27\,GB & 46.74\,G \\
32768   & 5.62 & 5.88 & 3.06 & 2.24\% & 1.99\,h & 0.006\,s & 0.52\,GB & 93.48\,G \\
100000  & \textbf{5.56} & 5.90 & 3.22 & 2.25\% & 7.2\,h & 0.005\,s & 1.57\,GB & 285.29\,G \\
\bottomrule
\end{tabular}}
\end{table*}


\begin{table*}[htbp]
\centering
\caption{Discretization Invariance of baseline models and \ac{LRQ-Solver} on the 3D beam dataset. ``OOM'' denotes out-of-memory. Inference time is measured in milliseconds (ms) per sample. Training time is total wall-clock time in hours.}
\label{tab:beam_mae_various_pts}
\resizebox{\linewidth}{!}{
\begin{tabular}{lcccccccccc}
\hline
\textbf{Model} & \textbf{Metric} & \multicolumn{7}{c}{\textbf{Point Numbers}} & \textbf{Epochs} & \textbf{Training Time (h)} \\
 & & 250 & 1k & 2k & 5k & 10k & 25k & Full Volume & & \\
\hline
\multirow{2}{*}{DeepONet} 
    & MAE       & 7.14 & 7.14 & 7.14 & 7.14 & 7.14 & 7.14 & 7.16 & \multirow{2}{*}{150\,000} & \multirow{2}{*}{27.5} \\
    & Time (ms) & \textbf{2.1}  & \textbf{2.2}  & \textbf{2.3}  & \textbf{2.4}  & \textbf{2.5}  & \textbf{2.6}  & \textbf{2.7}  &        &  \\
\hline 
\multirow{2}{*}{RegDGCNN} 
    & MAE       & 129.8 & 96.08 & 49.48 & 34.75 & OOM & OOM & OOM & \multirow{2}{*}{2\,000} & \multirow{2}{*}{18.4} \\
    & Time (ms) & 1\,101.6 & 2\,494.0 & 4\,823.4 & 16\,836.6 & — & — & — &  &  \\
\hline
\multirow{2}{*}{Transolver} 
    & MAE       & 41.14 & 38.84 & 38.08 & 37.95 & 37.79 & 37.70 & 37.69 & \multirow{2}{*}{2\,000} & \multirow{2}{*}{13} \\
    & Time (ms) & 537.6 & 901.7 & 1\,586.4 & 3\,348.3 & 6\,572.6 & 15\,673.0 & 37\,097.0 &  &  \\
\hline
\multirow{2}{*}{GeomDeepONet} 
    & MAE       & 2.33 & 2.33 & 2.33 & 2.33 & 2.33 & 2.32 & 2.32 & \multirow{2}{*}{150\,000} & \multirow{2}{*}{46.9} \\
    & Time (ms) & 2.7  & 2.7  & 2.7  & 2.7  & 2.7  & 2.7  & 2.7  &         &  \\
\hline
\multirow{2}{*}{LRQ-Solver (Ours)}  
    & MAE       & \textbf{1.66} & \textbf{1.65} & \textbf{1.66} & \textbf{1.66} & \textbf{1.66} & \textbf{1.66} & \textbf{1.66} & \multirow{2}{*}{2\,000} & \multirow{2}{*}{0.76} \\  
    & Time (ms) & 3.5  & 3.5  & 3.5  & 3.5  & 3.5  & 3.5  & 3.5  &       &  \\
\hline
\end{tabular}
}
\end{table*}

Robustness to geometric discretization is a key requirement for neural PDE solvers in industrial applications, where simulations must operate reliably across multi-fidelity meshes—from coarse design prototypes to high-resolution validation models. We evaluate this property on two representative problems with distinct physical characteristics: the turbulent flow field around a vehicle (\textit{DrivAerNet++}) and the smooth stress distribution in a structural beam (\textit{3D Beam}). The results reveal how \ac{LRQ-Solver} adapts its predictive behavior to the underlying physics, maintaining high fidelity and efficiency across resolution scales.

On the \textit{DrivAerNet++} dataset, the flow field exhibits strong spatial gradients, boundary layers, and wake structures—features that benefit from higher point density. As shown in Table~\ref{tab:point_num_exploration_drivaerpp}, \ac{LRQ-Solver} achieves progressively better accuracy as resolution increases, reaching an MSE of $5.56 \times 10^5$ and MRE of 2.25\% at 100k points. This gradual improvement reflects the model’s ability to resolve fine-scale flow details with more data. Notably, even at very low resolutions (1k–4k points), the MRE remains below 2.37\%, indicating strong generalization and effective feature extraction from sparse inputs. Training time and memory grow sub-linearly, while inference latency stays below 8 ms across all scales, with most configurations running in just 5 ms.

In contrast, the \textit{3D Beam} dataset features a smoother, more globally coherent stress field governed by linear elasticity. Here, the optimal physical representation can be captured at low resolution, and additional points provide diminishing returns. As shown in Table~\ref{tab:beam_mae_various_pts}, \ac{LRQ-Solver} achieves near-perfect discretization invariance: the MAE stabilizes at 1.66 MPa from 250 to over 35,000 points, with no performance drift. This flat error curve demonstrates that the model learns a resolution-agnostic representation early and maintains it consistently—ideal for applications involving heterogeneous or adaptive meshing.

The stark difference in convergence behavior between the two datasets highlights a key strength of \ac{LRQ-Solver}: it does not impose a fixed inductive bias toward overfitting local neighborhoods. Instead, through the covariance-based low-rank attention and parameter-conditioned modeling, it focuses on global physical coherence. In turbulent flows, it leverages higher resolution to refine local gradients; in smooth fields, it avoids unnecessary complexity and preserves stability.

Moreover, \ac{LRQ-Solver} maintains consistent inference latency across both datasets—around 3.5–5 ms regardless of input size—while baselines like RegDGCNN and Transolver suffer from rapidly increasing runtime. This efficiency, combined with adaptive resolution handling, enables seamless deployment in real-world design workflows involving iterative optimization, multi-fidelity simulation, and geometry variation.

These results confirm that \ac{LRQ-Solver} is not only accurate and fast but also physically aware: it understands when more data matters and when it doesn’t, adapting its behavior to the nature of the physical field. This makes it uniquely suitable for general-purpose industrial simulation across diverse problem types.

We evaluate the robustness of \ac{LRQ-Solver} to point cloud resolution across a range of discretization densities from 250 to 25k points. As shown in Table~\ref{tab:beam_mae_various_pts}, our model maintains consistent accuracy with MAE stable around 1.66 MPa across all resolutions, demonstrating strong invariance to sampling density. In contrast, RegDGCNN and Transolver degrade as point count increases. \ac{LRQ-Solver} also maintains stable inference latency across resolutions, unlike RegDGCNN, whose runtime increases sharply with more points. This robustness enables reliable deployment in practical engineering scenarios where geometry representations vary widely. These results confirm that \ac{LRQ-Solver} generalizes well across discretization levels.

\subsection{Visualization}
\label{sec:visualization}

\begin{figure*}[htbp]
    \centering
    \includegraphics[width=0.9\textwidth]{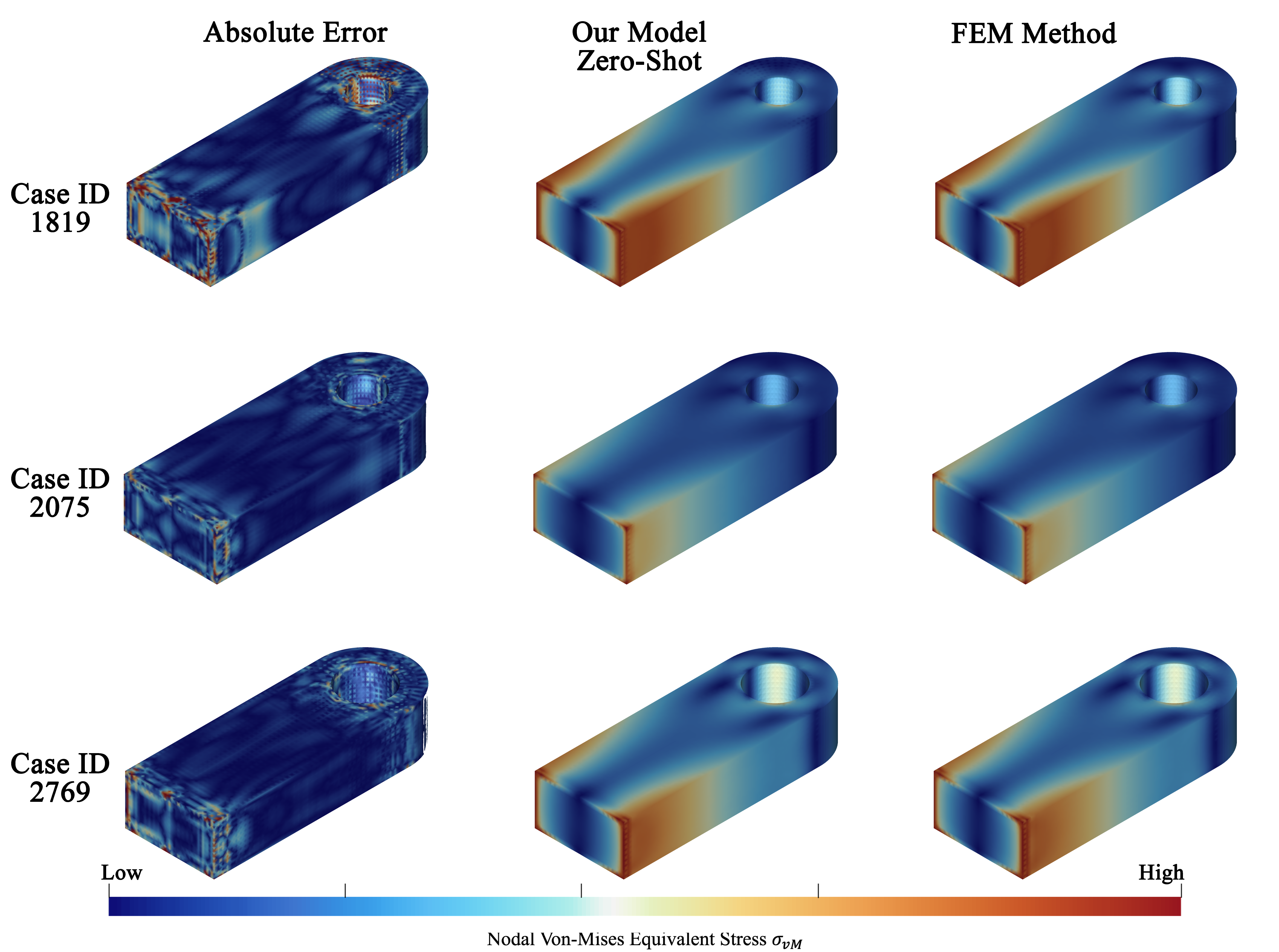}
    \caption{\textbf{Zero-Shot Model Prediction Result}: Our model predicts the nodal equivalent Von-Mises stresses over some case. The color scale indicates stress magnitude, with red regions corresponding to high stress concentrations near the hole.}
    \label{fig:beam_post}
\end{figure*}

\begin{figure*}[htbp]
    \centering
    \includegraphics[width=0.9\textwidth]{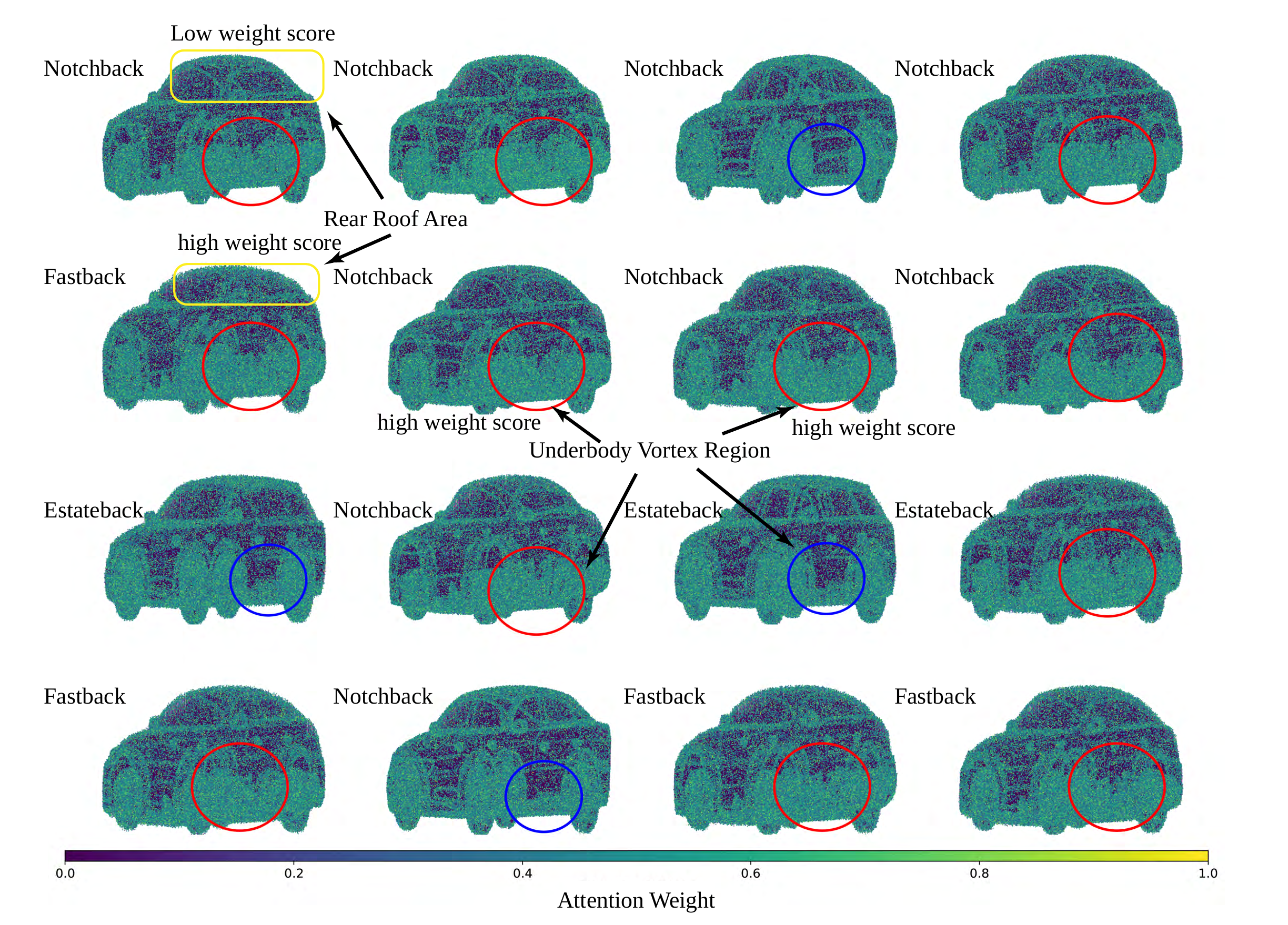}
    \caption{\textbf{Attention Heatmap Visualization}: Learned attention weights of \ac{LR-QA} during inference on a \textit{DrivAerNet++} sample. Warmer colors indicate stronger attention, revealing long-range interactions in the wake and local correlations near the surface.}
    \label{fig:drivaerpp_attention}
\end{figure*}

We visualize the predictive capabilities of \ac{LRQ-Solver} through two representative examples: stress field prediction on the \textit{3D Beam} dataset and attention mechanism analysis on the \textit{DrivAerNet++} dataset. 

Fig~\ref{fig:beam_post} visualizes the predicted nodal von Mises stress distribution for Test Cases on the \textit{3D Beam} dataset, alongside the \ac{FEM} ground truth and absolute error map. The color-coded cloud clearly reveals high-stress regions concentrated around the circular hole, consistent with classical mechanics predictions of stress concentration due to geometric discontinuity. The smooth gradient from the loaded region to the fixed end and sharp peak near the hole edge demonstrate that \ac{LRQ-Solver} accurately captures both global load transfer and local stress singularity. The absolute error map shows minimal deviation from \ac{FEM} results, with errors primarily localized near the hole boundary—expected due to stress gradients. This close agreement confirms that \ac{LRQ-Solver} produces physically plausible and accurate predictions in structural mechanics, capable of zero-shot generalization to unseen configurations without retraining.

Fig~\ref{fig:drivaerpp_attention} visualizes the learned attention weights of \ac{LR-QA} during inference on a \textit{DrivAerNet++} sample, with warmer colors indicating stronger attention. The heatmaps reveal that the model focuses primarily on regions critical to aerodynamic performance, particularly the rear end and wake area. Strong attention is observed around the trailing edge, roofline, and rear bumper—key locations where flow separation and pressure recovery occur. These areas directly influence the pressure drag component, which dominates total drag for ground vehicles. Additionally, local correlations near the surface suggest the model captures boundary layer behavior and skin friction effects. This spatial pattern aligns with fluid dynamics principles: the rear geometry governs wake structure and pressure distribution, while surface features affect flow attachment and turbulence. The consistent focus on these regions across multiple samples demonstrates that \ac{LRQ-Solver} learns physically meaningful attention patterns, enabling accurate prediction of drag coefficients by prioritizing the most influential geometric features.

\section{Conclusion}
\label{sec:conclusion}
We present \ac{LRQ-Solver}, a transformer-based neural operator for fast and accurate \ac{PDEs} Solving on complex 3D geometries at scale. To boost prediction accuracy across diverse design configurations, we introduce Parameter-Conditioned Lagrangian Modeling (\ac{PCLM}), which explicitly conditions local physical states on global parameters, enhancing physical consistency and reducing generalization error. To enable extreme computational efficiency, we propose covariance-based low-rank attention (\ac{LR-QA}), which reduces attention complexity from $O(N^2)$ to $O(NC^2 + C^3)$ by exploiting field covariance structure, eliminating point-wise clustering while preserving global coherence. Together, these innovations allow \ac{LRQ-Solver} to handle up to 2 million points on a single GPU, achieving a 38.9\% error reduction on \textit{DrivAerNet++} and 28.76\% on \textit{3D Beam}, with up to 50$\times$ faster training---demonstrating state-of-the-art accuracy, scalability, and efficiency for \ac{PDEs} Solving. The model exhibits strong discretization invariance and robustness to resolution and geometry variations, making it ideal for real-world engineering workflows. By embedding physics into the transformer backbone, \ac{LRQ-Solver} moves beyond black-box approximation, establishing a scalable, design-aware, and physics-informed paradigm for industrial-grade \ac{PDEs} Solving.

\section*{Acknowledgment}
The comparative models used in this study are publicly available and used in compliance with their respective licenses and ethical standards. We adhere rigorously to established research ethics throughout our work. As for the AI assistant, we utilize Qwen to identify textual errors and polish our paper and code.

\bibliographystyle{IEEEtran}
\bibliography{main}

\begin{IEEEbiography}[{\includegraphics[width=1in,height=1.25in,clip,keepaspectratio]{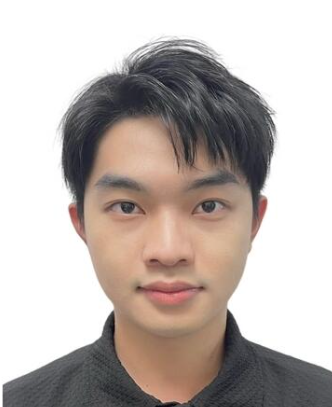}}]{Peijian Zeng}
received the B.S. degree in mechanical design, manufacturing and automation from Zhaoqing University, Zhaoqing, China, in 2018, and the M.S. degree in computer science and technology from Guangdong University of Technology, Guangzhou, China, in 2022. He is currently pursuing the Ph.D. degree in computer science and technology at Guangdong University of Technology. His research focuses on the application of artificial intelligence in computational mechanics and computational fluid dynamics.
\end{IEEEbiography}

\begin{IEEEbiography}[{\includegraphics[width=1in,height=1.25in,clip,keepaspectratio]{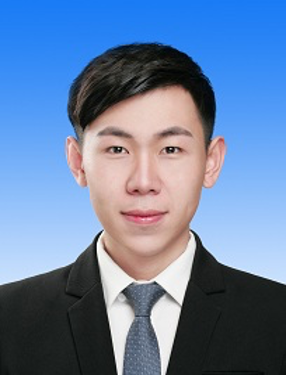}}]{Wang Guan}
received the B.S. in Naval Architecture from Harbin Institute of Technology, China in 2014, and the M.S. in Computational Mechanics from École Centrale de Nantes, France in 2020. He is currently a Senior Algorithm Engineer at Baidu, researching AI for Computational Fluid Dynamics and Partial Differential Equations.
\end{IEEEbiography}

\begin{IEEEbiography}[{\includegraphics[width=1in,height=1.25in,clip,keepaspectratio]{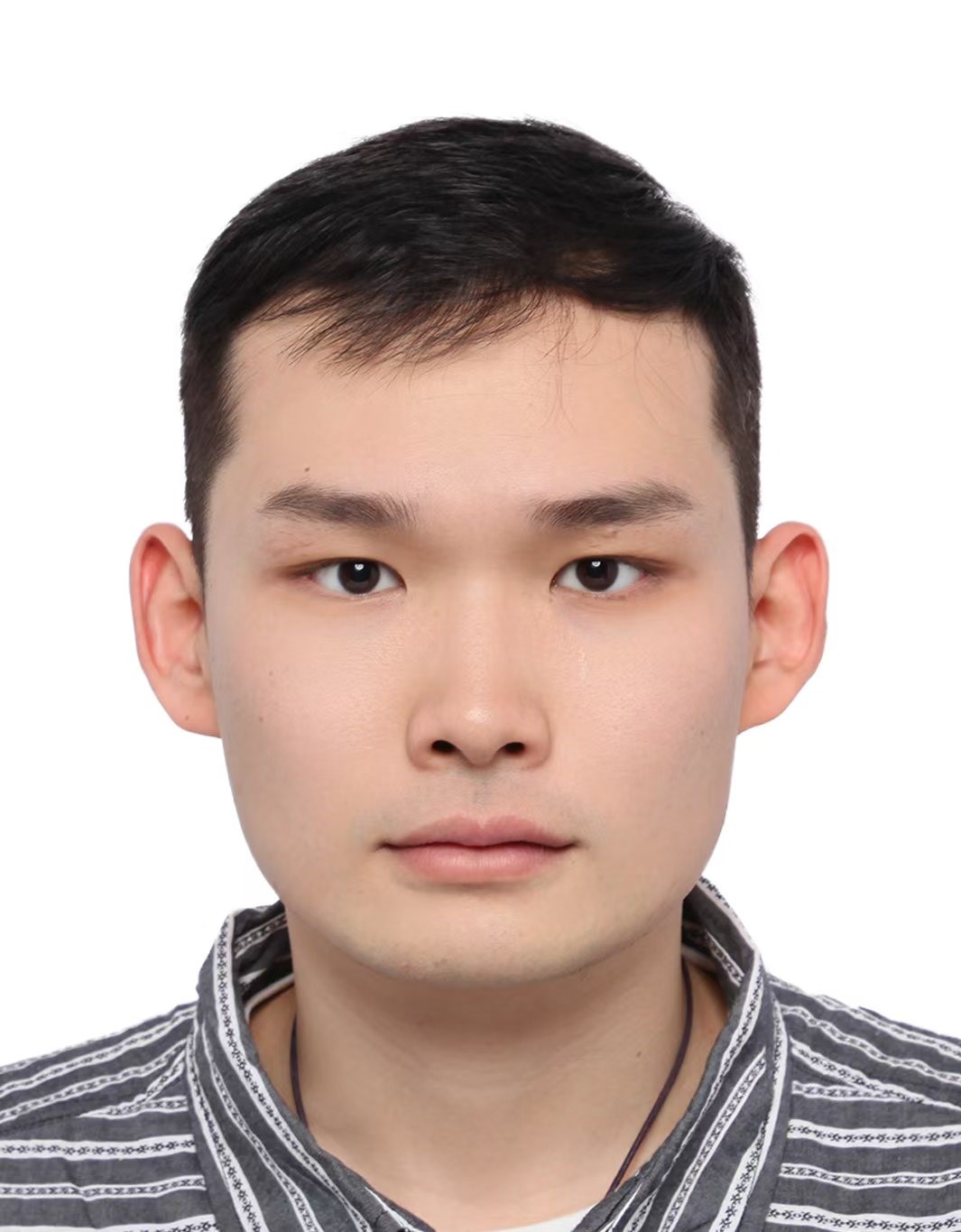}}]{Haohao Gu}
received the B.S. degree in Mechanics (Energy and Resource Engineering) from the Peking University, Beijing, China in 2019, and the Ph.D. degree in Mechanics from the Peking University, Beijing, China in 2024. He is currently an algorithm engineer with Beijing Baidu Netcom Science Technology Co., Ltd. His research interests focus on Artificial Intelligence for Energy Science and Engineering, Deep Generative Models.
\end{IEEEbiography}

\begin{IEEEbiography}[{\includegraphics[width=1in,height=1.25in,clip,keepaspectratio]{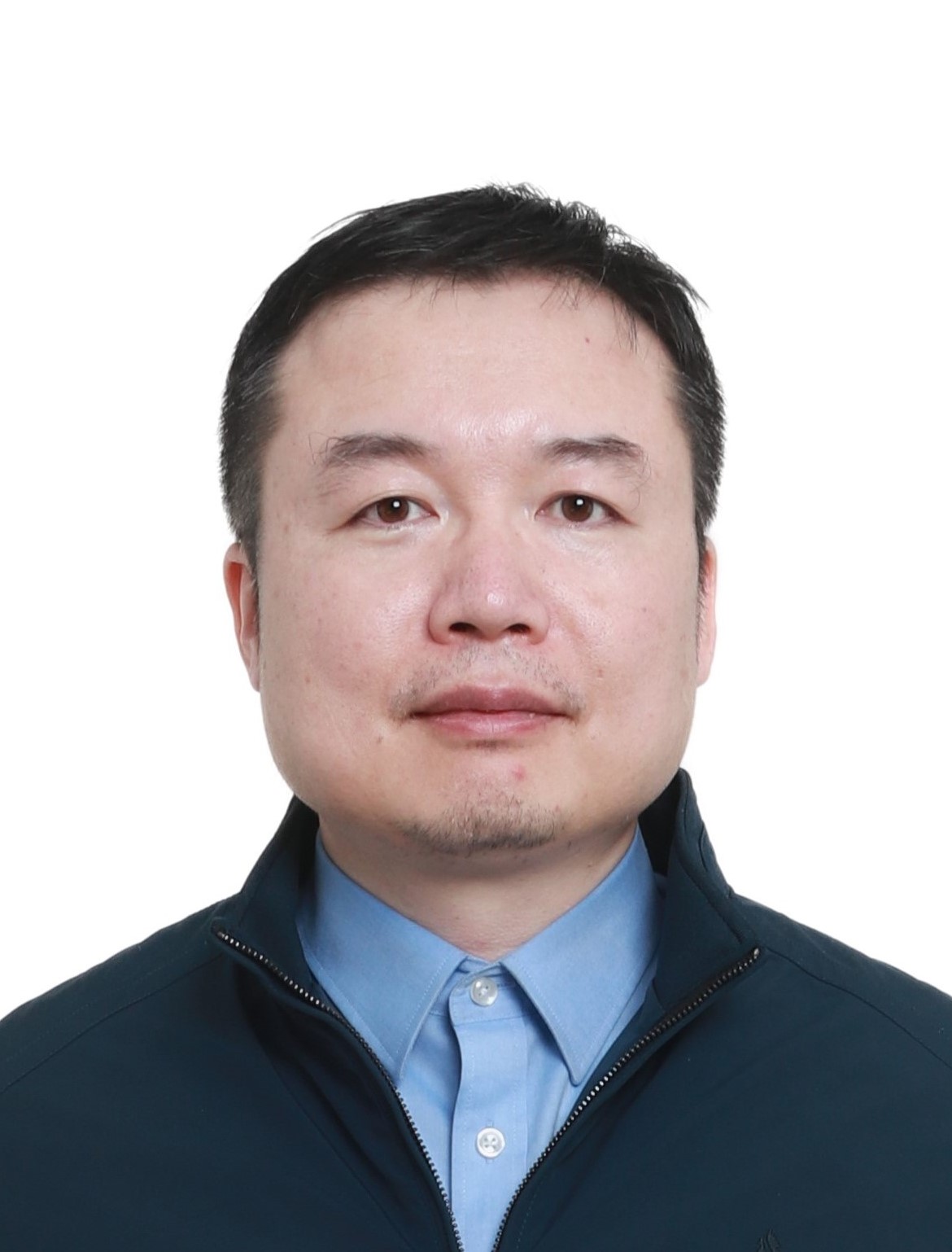}}]{Xiaoguang Hu}
received his Master's degree in Computer Science from Harbin Institute of Technology in 2006. He currently serves as a Distinguished Architect at Baidu, with research interests encompassing Natural Language Processing (NLP), Deep Learning Frameworks, and AI for Science.
\end{IEEEbiography}

\begin{IEEEbiography}[{\includegraphics[width=1in,height=1.25in,clip,keepaspectratio]{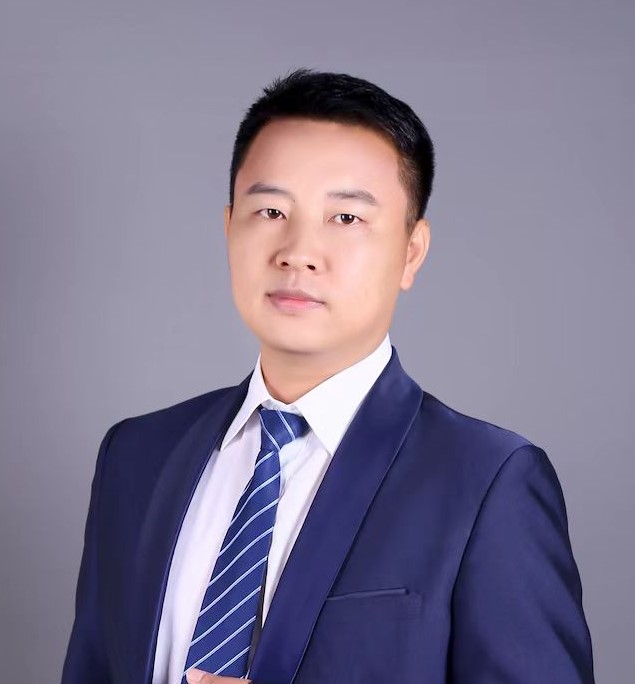}}]{Tiezhu Gao}
received the M.S. degree in Computer Science and Technology from Harbin Engineering University in 2008. He currently serves as a Senior Technical Manager at Baidu, where he leads research and development efforts on the PaddlePaddle deep learning framework.
\end{IEEEbiography}

\begin{IEEEbiography}[{\includegraphics[width=1in,height=1.25in,clip,keepaspectratio]{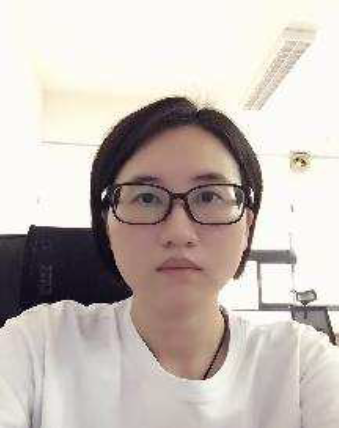}}]{Zhuowei Wang}
received the B.S. degree in com-puter science and technology from the China Univer-sity of Geosciences, Wuhan, China, in 2007, and the M.S. and Ph.D. degrees in computer systems archi-tecture from Wuhan University, Wuhan, in 2009 and 2012, respectively. From 2019 to 2020, she worked as a Visiting Scholar with the Norwegian University of Science and Technology, Gj\o{}vik, Norway. She is currently a Professor with the School of Computer Science and Technology, Guangdong University of Technology, Guangzhou, China. Her research interests focus on high-performance computing, low-power optimization, and distributed systems.
\end{IEEEbiography}

\begin{IEEEbiography}[{\includegraphics[width=1in,height=1.25in,clip,keepaspectratio]{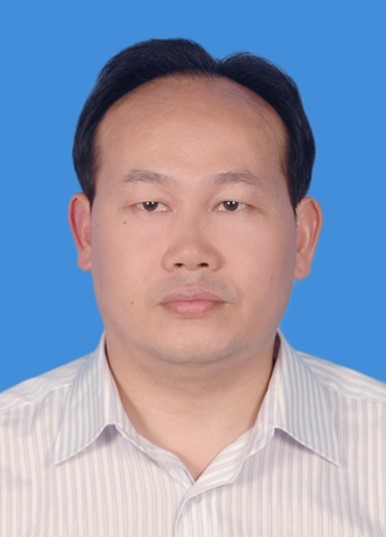}}]{Aimin Yang}
received the B.S. degree in physics from Hunan University of Science and Technology, Xiangtan, China, in 1993, the M.S. degree in computer science from the National University of Defense Technology, Changsha, China, in 2001, and the Ph.D. degree in computer software from Fudan University, Shanghai, China, in 2005. He is currently a Professor with the School of Computer Science, Guangdong University of Technology, Guangzhou, China. His research interests include the application of artificial intelligence in structural dynamics and natural language processing.
\end{IEEEbiography}

\begin{IEEEbiography}[{\includegraphics[width=1in,height=1.25in,clip,keepaspectratio]{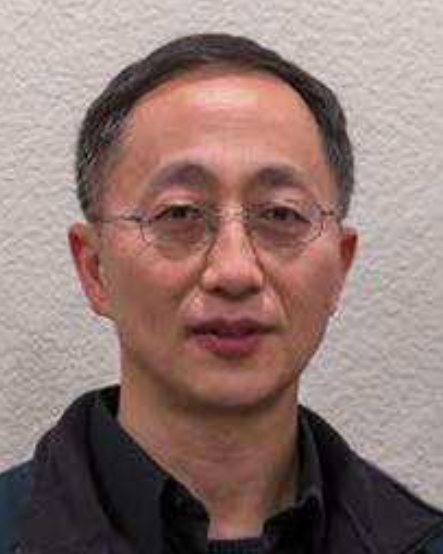}}]{Xiaoyu Song}
received the Ph.D. degree from the University of Pisa, Italy, in 1991. From 1992 to 1998, he was on the faculty at the University of Montreal, Canada. He joined the Department of Electrical and Computer Engineering at Portland State University in 1998, where he is now a Pro-fessor. He was an editor of IEEE Transactions on VLSI Systems and IEEE Transactions on Circuits and Systems. He was awarded an Intel Faculty Fellowship from 2000 to 2005. His research interests include formal methods, design automation, embedded systems and emerging technologies.
\end{IEEEbiography}

\end{document}